\documentclass[12pt]{amsart}
\usepackage{latexsym,fancyhdr,amssymb,color,amsmath,amsthm,graphicx,listings,comment}
\usepackage[section]{placeins}
\newtheorem{thm}{Theorem}
\newtheorem{propo}{Proposition}

\newtheorem{coro}{Corollary}
\let\paragraph\subsection
\title{Cohomology of open sets}

\author{Oliver Knill}
\date{Mai 21, 2023}
\address{Department of Mathematics \\ Harvard University \\ Cambridge, MA, 02138 }


\begin{document}
\maketitle

\begin{abstract}
If $G$ is a finite abstract simplicial complex and $K \subset G$ is a subcomplex and
$U=G \setminus K$ is the open complement of $K$ in $G$, the Betti vectors of $K$ and $U$ 
and $G$ satisfy $b(G) \leq b(K)+b(U)$.
\end{abstract}

\section{The inequality}

\paragraph{}
For a {\bf finite abstract simplicial complex} $G$ with $n$ elements, the {\bf exterior derivative}
$d$ is a $n \times n$ matrix, fixed once $G$ and each $x \in G$ are ordered.
The set of non-empty sets $G$ carries a finite non-Hausdorff {\bf Alexandroff topology} 
$\mathcal{O}$. Sub simplicial complexes are the closed sets and
{\bf stars} $U(x) = \{ y, x \subset y\}$ are the smallest open sets. 
A set $x$ of cardinality $p+1$ is a {\bf $p$-simplex}. The set of {\bf forms} $G \to \mathbb{R}$, 
identifies with $\mathbb{R}^n$. Functions on $p$-simplices are {\bf p-form},
$d$ maps $p$-forms to $(p+1)$-forms and the transpose $d_p^*$ maps $(p+1)$-forms 
to $p$-forms. For example, $d_0$ is the gradient, $d_1$ the curl, 
$d_0^*$ the divergence and $d_0^* d_0$ the {\bf Kirchhoff matrix}.

\paragraph{}
The {\bf cohomologies} $H(K),H(U)$ are defined Hodge theoretically:
for the open $U$ and closed $K=G \setminus U$, there are $|U| \times |U|$ or $|K| \times |K|$ 
matrices $d_U$ or $d_K$, so that $d_U \oplus d_K$ is a $n \times n$ matrix again. Define the
{\bf Dirac matrix} $D=d+d^*$ and a {\bf Hodge Laplacian} $L=D^2=(d+d^*)^2=d d^* + d^*d$. 
It is block diagonal $L=L_0 \oplus L_1 \oplus \cdots \oplus L_q$ and the {\bf Betti vector}
$b = (b_0,b_1, \dots, b_q)$ with {\bf Betti numbers} 
$b_p = {\rm dim}({\rm ker}(L_p))$, where $q$ is the {\bf maximal dimension} of $G$. 
The space of {\bf harmonic $p$-forms} ${\rm ker}(L_p)$ is the {\bf $p$-th cohomology}.
$H(U)={\rm ker}(L_p(U))$ can naturally be identified with the {\bf reduced relative cohomology} $H(G,K)$
and $H(K)={\rm ker}(L_p(K))$ with $H(G,U)$. {\bf Excision} $H(U)=H(G\setminus V,K \setminus V)=H(G,K)$ holds.

\paragraph{}
The blocks $L_p(U)$ for open sets $U$ can also be $0 \times 0$
matrices. It happens if $U$ has no $p$-simplices. For the open $U = \{x\}$ for example, 
where $\{x\}$ is a maximal $q$-simplex in $G$, the matrix 
$L_q(U)$ is a $1 \times 1$ matrix, while the other other blocks $L_p(U)$ are all $0 \times 0$ matrices
so that $b_q(U)=1$ and all other $b_p(U)=0$. This is the cohomology of every open $q$-ball. 
For a non-negative integer vector $b$, there is a complex $G$ and an open $U \subset G$
so that $b(U)=b$. Also any $f$-vector $f=(f_0,f_1,\dots)$, with $f_k$ counting $k$-simplices
can be realized for open sets. For closed $K$, this is not true due to 
{\bf Kruskal-Katona} constraints. 

\paragraph{}
If $b(U),b(K)$ and $b(G)$ are the Betti vectors of $U,K$ and $G$,
the {\bf Euler-Poincar\'e formula} $\chi = \sum_{p=0}^q (-1)^p f_p = \sum_{p=0}^q (-1)^p b_p$
for {\bf Euler characteristic} $\chi$ holds both for open and closed sets:
apply the {\bf heat flow} $e^{-t L}$, using the {\bf McKean-Singer symmetry} 
stating that ${\rm str}(e^{-t L})$ is independent of $t$.
So, ${\rm str}(e^{-t \cdot L})= \sum_{p=0}^q (-1)^p f_p$ for small $t \geq 0$ 
and ${\rm str}(e^{-t L}) = \sum_{p=0}^q (-1) b_p$ for large $t \geq 0$. Simple counting gives $f(G)=f(K)+f(U)$. 
For cohomology however:

\begin{thm}[Fusion inequality]
\label{1}
$b(G) \leq b(K)+b(U)$. 
\end{thm}

\begin{proof} 
Like $L_G$, also the direct sum $L_{U,K}=L_U \oplus L_K$ of the disjoint union $K \cup U$ is a $n \times n$ matrix.
We will show that the eigenvalues satisfy $\lambda_k(L_{K,U}) \leq 2 \lambda_k(L_G)$. 
This implies that $L_G$ can not have more zero eigenvalues 
than $L_{K,U}$. When restricting this argument to $p$-forms, 
the $p$-form Laplacian  $L_{G,p}$ can not have more zero eigenvalues than $L_{K,U,p}$,
proving that $b_p(G) \leq b_p(K) + b_p(U)$, which is the statement needing verification. \\
We know already $\lambda_k(L_U) \leq \lambda_k(L_G)$ and 
$\lambda_k(L_K) \leq \lambda_k(L_G)$ \cite{Eigenvaluebounds}. 
Due to block diagonal nature of $L_{K,U} = L_K \oplus L_U$ and $D_{K,U} = D_K \oplus D_U$,
the two matrices $D_K,D_U$ commute. The {\bf Courant-Fisher formula} for the eigenvalues with 
$\mathcal{S}_k(A) = \{ \{ |v|=1 \} \subset V, V \in \mathcal{G}_k(A) \}$ being the {\bf unit spheres}
of the points in the {\bf Grassmannian} $\mathcal{G}_k(A) = \{ V \subset \mathcal{R}^{|A|}, {\rm dim}(V)=k \}$
$$ \lambda_k(L_{K,U}) = \underset{j \geq 1}{{\rm min}}  \; \; 
                        \underset{V \in \mathcal{S}_j(U)}    {{\rm min}} \; \; 
                        \underset{W \in \mathcal{S}_{k-j}(K)}{{\rm min}}  \; \; 
                        \underset{v\in V, w \in W}{{\rm max}} (\lvert D_U v\rvert^2 + \lvert D_K w\rvert^2) \; . $$
This is smaller or equal than the sum of 
$\lambda_j(L_{K}) +\lambda_{k-j}(L_U)$ and so smaller or equal than $\lambda_k(L_K) + \lambda_k(L_U)$ as $j \leq k$ implies 
$\lambda_j \leq \lambda_k$ and $k-j \leq k$ implies $\lambda_{k-j} \leq \lambda_k$. 
Both are estimated from above by $\lambda_k(L_{G}) = \underset{V \in \mathcal{S}_k(U)}{{\rm min}} \underset{v\in V}{{\rm max}} \lvert D_G v \rvert^2$.
Adding up the two estimates gives the result. 
\end{proof}

\paragraph{}
This proof was based on three pillars: \\
{\bf 1)} The cohomology is a spectral property as it was defined as such. In the case
of closed sets, this was equivalent to the simplicial cohomology and is classical 
(as has been noted by \cite{Eckmann1944}). In the case of open sets, this appears to be 
a new definition. \\
{\bf 2)} We use spectral monotonicity with respect to the lattice of closed subsets $K$. This was
proven earlier \cite{HodgeInequality} and was based on Courant-Fischer formula.
An appendix provides code, allowing an interested reader to compute all
objects, vectors, matrices and numbers discussed here for an arbitrary simplicial complex $G$.\\
{\bf 3)} If a symmetric matrix $L=A \oplus B$ decays into two block diagonal matrices 
$A,B$, then the left padded increasing order spectra satisfy 
$\lambda_k(A) + \lambda_k(B) \leq 2 \lambda_k(A \oplus B)$. 
Some elementary properties of a {\bf partial order on finite sequences} or 
the {\bf spectral partial order on matrices} in an appendix.

\paragraph{}
Computing the cohomology of open sets in the same time also determines the {\bf relative 
cohomology} $H(U)=H(G,K)$ which in turn is the {\bf reduced cohomology} $H(G/K)$ of the quotient $G/K$. 
The later is not a simplicial complex in general but it can be identified with 
a $\Delta$-set. Any result which can be formulated for open sets like Gauss-Bonnet, 
Poincar\'e-Hopf or the already mentioned Euler-Poincar\'e or 
McKean-Singer results can be reinterpreted as results for $\Delta$-sets and so in particular 
for simplicial sets which are a subclass of $\Delta$ sets with more
additional structure. Every $\Delta$-set can be obtained by applying 
operations $G \to G/K$ and taking as the topology on $X=G/K$
the open sets in $U$, together with $X=U \cup \{x\}$, where $x$ represents the 
equivalence class of $K$. If $(G,\mathcal{O})$ was a connected 
topological space, then $(X=G/K=U \cup \{x\},\mathcal{O} \cup \{X\})$ is connected. 
The only new open set is $X$. The only new closed set is $\{x\}$. In other words,
$X$ is the {\bf one-point compactification} $\dot{U}$ of $U$. We could keep it a simplicial 
complex by making a {\bf cone extension} over the boundary of $U$ in $G$; it is much cheaper
to just compute the cohomology of $U$ however. The spectral properties of $L_{\dot{U}}$ via
cone would of no more allow a direct comparison of the spectrum of the extension with the 
spectrum of $L_G$.The closure $\overline{U}$ of $U$ would be a subcomplex of $G$ but the cohomology of 
$\overline{U}$ and $\dot{U}$ are not related in general as the case $G=\overline{U},
K=\delta{U}$ shows. By the way, experiments indicate that there are much more sets
without interior than closed sets that are closures of open sets. On the complete
complex $G=K_{q+1}$ (which can be seen as the smallest closed $q$ dimensional ball) 
for example, every open set in $G$ is dense in $G$, illustrating the non-Hausdorff property. 

\paragraph{}
Also {\bf branched coverings} can be dealt with more elegantly when including open sets because one can 
separate the closed {\bf ramification set} from the rest, which is an open set. An example is a
wedge of of $m$ circles which can be seen as a branched cover with a single ramification locus. 
If a group $A$ acts on a simplicial complex $G$ and $K$ is the orbit of a single point, then
the equivalence classes $G/A$ of orbits can be seen as $G/K$. Its {\bf reduced cohomology} $H(G,K)$ is 
defined as the cohomology of the open set $U=G \setminus K$. In the case where $G$ was a circle and $K$ a set of 
$m$ different points, the set $U$ is a cover of the {\bf fundamental domain} $V$ and consists of 
$m$ disjoint unions of open sets. Of course, $b(V)=(0,1)$ and $b(U)=(0,m)$ and the cohomology of
the sphere bouquet is $b(G) = b(K) + b(U) - b(I) = (m,m) - (m-1,m-1) = (1,1)$. In general, for any 
complex $G$, if $K$ is a $0$-dimensional subcomplex consisting of $m$ points, $b(G)=b(U)+b(K)-b(I)$ with 
$b(()=(m-1,m-1,0,\dots)$. An alternative {\bf Riemann-Hurwitz picture} for a group $A$ acting on a complex $G$ 
is to look at the {\bf fixed point set} $K$ (the ramification points in a branched cover setting) separately. 
Let $G$ be the bouquet of $m$ circles and $A$ the cyclic group operating petals. Let $K$ be the fixed point set. 
Now $U=G \setminus K$ is an open set with reduced cohomology $b(U) = (0,m)$ and Euler characteristic $-m$. 
We see $G$ as the 1-point compactification of $U=G \setminus K$ so that $b(G) = b(U) + b(K) = (0,m) + (1,0) = (1,m)$. 

\paragraph{}
This concludes the article. The rest is less condensed, contains more examples,
remarks and illustrations. We see a {\bf set of $n$ sets} $G$ with a {\bf dimension function} ${\rm dim}(x)$
monotonone in $|x|$ and for which the $n \times n$ matrix $D=d+d^*$ has a nilpotent $d$ as an effective
{\bf data structure} that represents elements in the {\bf elementary topos}
$\mathcal{G}$ of {\bf finite $\Delta$-sets}. The later is a {\bf functor category} and a {\bf presheaf} over
a simplex category: there is not only a notion of addition = {\bf coproduct}, the disjoint union, but
also a {\bf Cartesian product} and the possibility to look at equivalence classes $G/K$ for a
subobject $K$. The zero element $G=\{ \}$ is the {\bf initial object} and $G=\{1\}$ the 
{\bf terminal object}. The {\bf topos of finite sets} is not powerful enough for
{\bf calculus} and {\bf geometry}, $\mathcal{G}$ does the job. One has a multi-variable calculus in arbitrary
dimensions, there is a cohomology and Betti vectors $b$ which can be encoded in the Poincar\'e polynomial 
$b_G(t) = b_0 +b_1 t + \dots + b_q t^q$. The map $G \to b_G$ becomes a ring homomorphism if 
$\mathcal{G}$ is extended to a commutative ring $\mathcal{Z}$, which is a topos and presheaf.
K\"unneth follows by definition because the $p+q$ form $g(x) h(y)$ is a harmonic 
function in the product $G \times H$ if $g$ was harmonic in $G$ and $h$ was harmonic in $H$. The difficulty
which already Whitney battled with in 1950 that if the $p$-form $g$ is a function of $p+1$ variables and
the $q$ form $h$ is a function of $q+1$ variables, then $g(x) h(y)$ is a function of $p+q+2$ variables is 
taken care off because in $G \times H$, the $p+q$ forms are given as such because the dimension functional has shifted.
The same story worked already on the category of finite simple graphs which 
extends to a ring using the Shannon product and where the same K\"unneth relations assure that we have
a Betti functor from the ring of graphs to the ring of polynomials. While graphs are more intuitive there are
disadvantages: (i) computations become computationally heavier, (ii) we can not form $G/K$ with a subgraph $K$
in general, (iii) the complement $G \setminus K$ for a subobject $K$ is not in the class;
(iv) {\bf relative cohomology} are not convenient; and (v), the product of two manifolds
is not in general not a manifold. The $4$-manifold
$G=S^2 \times S^2$ can be represented in $\mathcal{G}$ as a set 
$G=\{\{1,5\},\{1,6,7,8\},\{2,3,4,5\},\{2,3,4,6,7,8\}\}$. Both the $f$-vector and
Betti vector are $(1,0,2,0,1)$. As the Kruskal-Kotona conditions illustrate, this
is far from a simplicial complex. The dimension functional has become important.
While in general, $D$ is already determined from $G$, we need the dimension function
to represent a topos element. 

\paragraph{}
The notion of manifold for simplicial complexes could be extended too from simplicial complexes to $\Delta$ set.
Every topos element $G$ comes with a finite topological space $\mathcal{O}$
generated by {\bf atoms} $U(x)$, smallest open sets in $\mathcal{O}$. The {\bf unit sphere} is 
$S(x)=\overline{U(x)} \setminus U(x)$, the {\bf unit ball} is $B(x) = \overline{U(x)}$. 
A {\bf $d$-manifold} is a space for which every unit sphere is a $(d-1)$-manifold that is a $(d-1)$ sphere and a $d$-sphere
is a $d$-manifold $G$ which when punctured $G\setminus U(x)$ becomes contractible. 
\footnote{Contractible is here understood as collapsible. We say {\bf homotopic to 1} 
for the wider equivalence relation, where one can do contraction and extension steps. Deciding
about homotopy equivalence is in other complexity class (provided $P \neq NP$) than deciding about contractibility.}
The topology on the 
$d$ sphere $G=\{ \{0\}, \{1,2,3, \dots, d\} \}$ for example consists of the open sets in $\{1, \dots, d\}$
(seen as a topological space itself when closing it). Now by definition, the unit sphere of $\{0\}$ is 
$\overline{U( \{0\})}  \setminus U( \{0\} )$ which is a $(d-1)$-sphere. 
If $M \subset G$ is a sub-object, the {\bf interior points} ${\rm int}(M) = \{ x \in M, U(x) \subset M \}$ 
(where $U(x)$  refers to the smallest open neighborhood of $x$ in $G$),
define the {\bf boundary} $\delta M =  M \setminus {\rm int}(M)$. A function $f: M \to \mathbb{R}$
a {\bf form}, is simply a vector indexed by the finite set $M$. 
Define the {\bf integral} $\int_M f dx = \sum_{x \in M} f(x)$. 
If the sets in $M$ can be oriented in a compatible way so that
if $x,y \in M$ and $x \subset y$ the orientations of $x,y$ match,
{\bf Stokes theorem} is $\int_M df dx  \int_{\delta M} f dx$ and if $M$ has no boundary, 
$\int_M df dx =0$. Traditionally, Stokes theorem is treated in the language of {\bf chains},
elements in the free Abelian group generated by the elements in $M$. On a single simplex just being the definition
of exterior derivative, it goes over to general chains by linearity. On orientable manifolds with 
boundary, the orientation produces cancellations of $df$ in the interior. Discrete calculus notions
have started more than 150 years ago \cite{Kirc}.

\section{Examples} 

\paragraph{}
An example, where Theorem~(1) is equality is if $U$ is an {\bf open $q$-ball} and $K$ a {\bf closed $q$-ball}.
If fused it becomes a $q$-sphere $G$. Then $b(U)=(0,\dots, 1), b(K)=(1,0,\dots,0)$ and $b(G)=(1,0,\dots,0,1)$.  
An example with inequality is if $U$ is an open $q$-ball and $K$ is a $(q-1)$-sphere
and $G=U \cup K$ is a closed $q$-ball. Then $b(U)=(0,0, \dots,0,1)$, $b(K)=(1,\dots, 1,0)$ and 
$b(G)=(1,0,\dots,0,0)$ so that the interface cohomology is $b(I) = b(U)+b(K)-b(G)=(0,0, \dots,0,1,1)$. 
IN the case $q=1$, we can implement this as $U=\{ \{1,2\} \}$, $K=\{ \{2\}, \{2,3\}, \{3,4\}, \{4\},\{4,1\}, \{1\} \}$
with $G=U \cup K$ being a 1-sphere. In this case $b(U)=(0,1), b(K)=(1,0), b(G)=(1,1)$.  

\paragraph{}
The {\bf comma space} $\{ \{1\}, \{1,2\} \}$, which can be seen as an open star in its closure
$\{ \{1\},\{1,2\},\{2\} \}$ we have $d=\left[ \begin{array}{cc} 0 & 0 \\ -1 & 0 \end{array} \right]$
and $L=1_2$. It is important that we do not need to know the ambient $G$ to get $d$. The object
$U$ itself is now considered a geometric object. The orientation (a choice of coordinates) only 
affects the entries of $D$ and not spectral properties. For
$\{ \{2\}, \{1,2\} \}$ for example, we would have 
$d=\left[ \begin{array}{cc} 0 & 0 \\  1 & 0 \end{array} \right]$.
The Hodge blocks are $L_0=L_1=[1]$. There are no harmonic forms in this
case. The Betti vector is $(0,0)$. Open sets with {\bf trivial cohomology}
are interesting when doing {\bf homotopy deformations}. If we add such an open set $U$ to 
a closed set $H$ to get a larger closed set $G$, there can be no exchange of cohomology 
and we must have additivity. The Fusion inequality must be an equality:
$b(H) + b(K) = b(G)$. Homotopy deformations do not change the cohomology. 
The same holds for any for any complete $q$-simplex in which one of $q-1$ dimensional
clsed face has been taken off, like $U(\{1\}) = \{ \{1 \},\{1,2\},\{1,3\},\{1,2,3\} \}$. In
this two-dimensional case, the Dirac matrix is $\left[ \begin{array}{cccc} 0 & -1 & -1 & 0 \\
                  -1 & 0 & 0 & 1 \\ -1 & 0 & 0 & -1 \\ 0 & 1 & -1 & 0 \\ \end{array} \right]$,
a matrix of determinant $4$ so that the cohomology is $b(U)=(0,0,0)$. The Laplacian $L=D^2$
is $2 I_4$. 

\paragraph{}
For a general set of sets $A$, like $A=\{\{1\},\{1,2\},\{1,2,3\}\}$
which need neither to be open or closed (and assuming that ${\rm dim}(x)=|x|-1$,  
the conditions $d^2=0$ already fail. We get a Hodge Laplacian that is no more block diagonal. 
The condition $d^2=0$ in $D=d+d^*$ is related to the axioms of the face maps if the object
is condered a $\Delta$-set. While Euler-Poincar\'e works in any situation where $d^2=0$, 
Euler-Poincar\'e fails in the example $A$ as $f_0=f_1=f_2=1$ but $b_0=b_1=b_2=0$.  
Euler-Poincar\'e would work for the closure 
$\overline{A}=\{\{1\},\{2\},\{1,2\},\{2,3\},\{1,3\},\{1,2,3\}\}$ or
the smallest open set $B=\{\{1\},\{1,2\},\{1,3\},\{1,2,3\} \}$
containing $A$, where $f_0-f_1+f_2=0$ and $b_0=b_1=b_2=0$.
Open sets are examples, where one still is in a $\Delta$-set situation, but which are no
more simplicial complexes. Also Cartesian products are only $\Delta$-sets. A good test
whether we deal with a $\Delta$-set is to see whether the exterior derivative 
$d$ defined from the set of sets satisfies $d^2=0$. In general, we have either implicitly 
like for open or closed sets in a simplicial complex also specify the {\bf dimension function}. 
Without saying otherwise, it is ${\rm dim}(x)=|x|-1$, the cardinality of $x$ minus $1$.  

\paragraph{}
Like in Mayer-Vietoris, we can use the decomposition to compute Betti vectors. A classical situation is the
{\bf connected sum construction} in which two complexes are identified along a closed ball $B$
removing the open interior of $B$ so that the complexes are glued along spheres. This does not need to be
a manifold situation. In general, if $M,H$ are two complexes and $x \in M, y \in H$ have isomorphic 
unit spheres $S(x),S(y)$, where $S(x)=B(x) \setminus U(x)$, with unit ball
$B(x)=\overline{U(x)}$ the closure of the smallest open set $U(x)$ containing $x$, then
we can look at the disjoint union $M' \cup H' = M \setminus U(x) \cup H \setminus U(y)$ and identify
them along $S(x) \sim S(y)$. This produces a new complex $G$. Mayer-Vietoris interprets this as 
$G=M' \cup H'$ with intersection $M' \cap H'$ which are all simplicial complexes. 
The new picture is to see $G$ as a union of a complementary pair of a {\bf closed}
$K=M \setminus U(x)$ and an {\bf open} $U=H = H \setminus B(x)$. 
If $b(K)$ and $b(U)$ are known, then $b(G) \leq b(K) + b(U)$. 
An interesting question is to give necessary or sufficient conditions under which conditions 
we have equality. The case when $b=0$ is an example where we have equality. 

\paragraph{}
Let us assume that $M,H$ are {\bf orientable $q$-manifolds}. {\bf Orientable} means that there is exists a 
non-zero $q$-form on $M$. The term {\bf q-manifold} means that every unit sphere 
$S(x)=\overline{U(x)} \setminus U(x)$ in $G$ is a $(q-1)$-sphere. Also inductively defined is the notion
of sphere. A {\bf $q$-sphere} $S$ is a $q$-manifold for which a sub-complex $A=S \setminus U(x)$ is
contractible. Contractible $A$ means that there is $y \in A$ such that both $S(y)$ and 
$G \setminus S(x)$ are contractible in $A$. Denote by $e_1, \dots e_q$ the standard basis vectors 
in the linear space $\mathbb{R}^q$. Mayer-Vietoris sees this as $b(M')=b(M) - e_q$ and $b(H')=b(H) -e_q$ 
and $b(G)=b(M')+b(H')-b(S')$, where $S'$ is the suspension of $S=M' \cap H'$ so that 
$b(G)=(1,b_1(M)+b_1(H),b_2(M)+b_2(H),\dots,b_{q-1}(M)+b_{q-1}(H),1)$.
Theorem~\ref{1} is here an equality. 

\paragraph{}
For example, if $M,H$ are both $2$-tori, then 
$b(M)=b(H)=(1,2,1)$ and while Mayer-Vietoris sees this as $b(M \# H) = (0,2,0) + (0,2,0) + (1,0,1) = (1,4,1)$ 
for the {\bf genus $2$ surface}, we see it as $b(K \cup U) = (1,2,0) + (0,2,1)= (1,4,1)$. 
If $M,H$ are both $2$-spheres, then $b(M)=b(H)=(1,0,1)$. We see it as $(1,0,0) + (0,0,1)=(1,0,1)$. 
More generally, we can add a topologically closed genus $k$ surface $K$ with $b(K)=(1,2k,0)$ and an 
open genus-$l$ surface $(0,2l,1)$ to a surface without boundary with $b(G)=(1,2k+2l,1)$. 

\paragraph{}
In the non-orientable case, we can have examples, where the connected sum operation produces a strict inequality. 
If $M,H$ are both projective planes for example, we see this as then 
$b(N \# H) = b(K \cup U)$, where $K$ is a closed ball and $U$ 
is an {\bf open Moebius strip} so that $(1,0,0) + (0,0,0)= (1,0,0)$. If we would take an open ball $U$ and a 
{\bf closed Moebius strip $K$} then $(0,0,1) + (1,1,0) \geq (1,0,0)$ produces a collision of a harmonic 
$1$ and $2$-form, when merged. 
Such a collision of a $1$-form and a $1$-form happens also if we fuse an open ball $U$ with a circle $K$ we get 
a $2$-ball $G$ and $b(U) + b(K) = (0,0,1) + (1,1,0) \geq (1,0,0) = b(G)$.

\paragraph{}
The {\bf join} $A * B$ of two simplicial complexes $A,B$ (in particular closed sets) 
is defined as the topological closure of the set of sets $\{ a*b, a \in A, b \in B\}$, where 
$a*b$ is the {\bf disjoint union} of two sets $a,b$. 
For example, $A=\{ \{1\},\{2\},\{3\}\}, B=\{\{4\},\{5\},\{6\} \}$ we take the 
closure of $\{ \{1,4\},\{1,5\},\{1,6\}, \{2,4\},\{2,5\},\{2,6\}, \{3,4\},\{3,5\},\{3,6\} \}$ which is the 
Whitney complex of the utility graph, a graph with $6$ vertices and $9$ edges, 
$f$-vector $f(G)=(6,9)$ and Betti vector $b(G)=(1,4)$. When seen in graph theory, the join was an ``addition" 
dual to the disjoint union. Here it closely related to multiplication $A \times B$. Indeed, for open sets
it is very close as for simplices $x * y$ is a simplex of dimension ${\rm dim}(x) + {\rm dim}(y)+1$. 

\paragraph{}
The join operation establishes a {\bf join monoid} structure on the set of
simplicial complexes. When considered for graphs, it is called the {\bf Zykov join} and 
dual to the disjoint union where duality is the {\bf graph complement}. 
This has been imprtant also in the arithmetic of graphs, where the Shannon ring with disjoint
union as addition and Shannon multiplication (the strong product) is an isomorphic dual ring 
to the Sabidussy ring with Zykov join addition and Sabidussy multiplication (large product). 
With the Whitney complex structure and cohomology the ring is compatible with cohomology and 
the map from the ring to the Poincar\'e-Polynomial is a ring homomorphism. 
It is not possible to extend the product to simplicial complexes in an associative way
because multiplication throws us out of the class of simplicial complexes so that we in the past
looked at the Barycentric refined object. But then $1*G=G_1$ is the Barycentric refinement of $G$
and $1*(1*G)=G_2$ the second Barycentric refinement while $(1*1)*G=G_1$ is the first. 
Multiplication throws us out from simplicial complexes into the larger class of $\Delta$ sets. 

\paragraph{}
The {\bf join of two open sets} can be defined  as $A * B = \{ a*b, a \in A, b \in B\}$ bt without 
shifing the dimension. The join of $A=\{ \{1,2\} \}$ and $B= \{ \{3,4\} \}$, both 1-dimensional 
objects is $A*B = \{ 1,2,3,4\}$ which is a $3$ dimensional object. This fits because the 
join of the $1$-point compactification of $A$ (a 1-sphere) and $B$ (an other 1-sphere) is now a 
3-sphere as it should be. The join is automatically open in any complex containing both $A$ 
and $B$ as disjoint sets. If we define for an open set $f_A(t) = \sum_{k=0}^q f_k(A) t^{k+1}$. 
For closed sets, we have $1+f_{A*B} = (1+f_A)(1+f_B)$ for open sets $A,B$, we have $f_{A*B} = f_A f_B$. 
With these definitions, $f_U(t) + f_K(t)  = f_G(t)$ if $G$ is considered closed.
If $A$ is a $k$-sphere and $B$ is a $l$-sphere, then $A+B$ is a $(k+l+1)$-sphere. To understand
the Cartesian product we have to see it as a Cartesian product of open sets then readjust the 
topology to make it closed, which involves shifting the dimensions. 
For example, an open $p$-ball multiplied with an open $q$-ball is an open $(p+q)$ ball with 
cohomology $(0,\dots,0,1)$. The Cartesian product of an open set behaves nicely.

\paragraph{}
The {\bf suspension operation} $G=S(H) = H*S_0$ is the join of $G$ with the
$0$-sphere $S_0=\{ \{1\},\{2\} \}$. In general, for any simplicial complex $G$, 
the Betti vector $b(G)$ is related to $b(H)$ by 
$$   b(S(H)) = T(b(H)-e_1) + e_1 \; ,  $$
where $T$ is the {\bf shift operator} $T(b)_k=b_{k+1}$. In other words, {\bf ``the suspension 
shifts the reduced cohomology"}. This formula is known in topology (see e.g. \cite{Hatcher})
and could be verified with a Mayer-Vietoris 
argument, seeing $S(G)$ as the union $A \cup B$ intersecting in $H$. Both $A,B$ are joins of 
$H$ with a $1$ point complex $1$ and have $b(A)=b(B)=e_1$. Mayer-Vietoris now sees
$b_k(G) = b_{k-1}(H)$ for $k \geq 2$ and $b_1(G)=b_1(H)$. 
We can also see it within $G$ as a disjoint union of $K=B(s_1)$ with $U=G \setminus K$. 
Now $b(K)=(1,0,0,0, \dots, 0)$ and $b(U)=(0,b_1,b_2,\dots,b_d)$. 
This again is a case with additivity $b(K)+b(U)=b(G)$. 

\paragraph{}
A {\bf homotopy extension} $K \to G$ can be understood as a fusion of a simplicial complex 
$K$ with an open set $U$ for which both the $1$-point compactification $\dot{U}$ and 
$\overline{U} \cap K$ are contractible. Since $b(\dot{U})=1 = 1+b(U)$, we must have $b(U)=0$. 
A {\bf homotopy deformation} is given by a finite set of such extensions or reductions, 
the inverse operation.  Mayer-Vietoris sees a homotopy extension as the union of 
$K$ with the ball $\overline{U}$ intersecting in $K \cap \overline{U}$. As
both $K \cap \overline{U}$ and $\overline{U}$ are contractible, the cohomology does not change. 
We reinterpret it again as a situation, where $b(K)+b(U)=b(G)$ and $b(U)=0$ is the zero vector. For example,
if $U = U(x)$ is a smallest open set in $G$ and {\bf the unit sphere} $S(x)=\overline{U(x)} \setminus U(x)$ 
is contractible and then $K=G \setminus U$ is a homotopy deformation of $G$. A simple example
is if $x=\{v\}$ is a vertex in $G$ such that $S(x)=\{y\}$ is a simplex in $G$. In that case,
$U(x)=S(x) + x \setminus S(x)$ is a {\bf cone extension} of the simplex $S(x)$ with $S(x)$ removed.
There is no cohomology on $U(x)$ (the Betti vector is the zero vector) and 
$\chi(U(x))=\chi(B(x)) -\chi(S(x)) = 1-1=0$. This is an example with equality. 

\paragraph{}
A simplicial complex $G$ defines an element $X$ in a polynomial ring with variables
in $V = \bigcup_{x \in G} x$. To do such an encoding  of $G$ in a 
{\bf Stanley-Reisner ring}, represent every $x=(v_1,\dots,v_k) \in G$
as a {\bf monomial} $X(x)=v_1 v_2 \cdots v_k$ so that $G$ is encoded in
$X=\sum_{x \in G} X(x)$. 
Given two complexes $G,H$ represented algebraically
as $X,Y$ define $XY$, label the monoids as $V$, define the graph $(V,E)$, where
$E$ are pairs $(a,b)$ in $V \times V$ such that either $a$ divides $b$ or $b$ 
divides $a$. The Whitney complex of this graph is the 
{\bf geometric product} $G \times H$ of $G$ and $H$. If $H=\{\{1\}\}$ the complex $G H$ is
the {\bf Barycentric refinement} of $G$. If $G$ is a $p$-manifold and $H$ is a $q$-manifold
then $G\times H$ is a $p*q$ manifold. For cohomology, there is the {\bf K\"unneth formula}:
define the {\bf Euler polynomial} $b_G(t) = b_0  + b_1 t + \cdots + b_q t^q$. Then $b_G(-1)$ is the
Euler characteristic. The {\bf K\"unneth formula} is
$b_{G \times H}(t) = b_G(t) b_H(t)$. 

\paragraph{}
The computation via the Stanley-Reisner ring can lead rather quickly to large matrices because in every
multiplication, we force it to become again a simplicial complex (by forming the order complex). 
It is better to see the product as the Cartesian product which is not a simplicial complex 
any more but can be dealt with as a $\Delta$-set. 
For example, $A=B= \{\{1\},\{2\},\{1,2\}\}$, then \begin{tiny}
$A \times B = \{\{1,3\},\{1,4\},\{2,3\},\{2,4\},\{1,2,3\},\{1,2,4\},\{1,3,4\},\{2,3,4\},\{1,2,3,4\}\}$
\end{tiny}
This is not a simplicial complex. The sets of cardinality $2$ are now the vertices. The order complex of
$A \times B$ is a simplicial complex. It follows quickly from the fact that if $f$ is harmonic in $G$
and $g$ is harmonic in $H$, then the {\bf tensor product} $f \times g(x,y) = f(x) g(y)$ is harmonic
in $G \times H$. It is much more convenient to let the product be a $\Delta$-set and not a simplicial complex. 
This allows to work with $|G| \cdot |H|$ sets rather than with the Barycentric refinement which 
has much more elements. Also the {\bf spectral properties} are like in the continuum. The eigenvalues
of $L_{G \times H}$ are  $\lambda_i(G) + \lambda_j(H)$. 

\paragraph{}
The Cartesian product for closed sets when extended to a product for open sets $U,V$.
can be seen within the frame work of $\Delta$-sets, which is an elementary topos and 
so Cartesian closed. When confining to simplicial complexes,
we would have to define $U \times V$ as the {\bf interior} of 
$\overline{U} \times \overline{V}$. But the product of any $\Delta$-set is defined. Indeed,
the category of $\Delta$-sets is a {\bf topos}, having all the nice properties like
products with terminal $1$, and coproducts with initial $0$, 
exponentials, which allows to represent {\bf graphs of functions} 
$\{ (X,Y), F(X) = Y \}$ as elements 
in the category and so define {\bf level surfaces} 
$G= \{X \in G^r, f(X_1, \dots, X_r)=0 \}$ 
in $G^r$ and a notion to decide whether we have a sub-objects or not. 
For now, we can  us {\bf fusion calculus} in the product 
for example, let the $1$-sphere $G$ the fusion of an open interval $U$ and a closed
interval $K$ with $b(G) = b(U) + b(K) = (0,1) + (1,0)$. Then
$b(G \times G) = (b(U) + b(K))*(b(U) + b(K)) = b(U) b(U) + 2 b(U) b(K) + b(K) b(K)
=(0,1)*(0,1) + 2 (0,1)*(1,0) + (1,0)*(1,0) = (0,0,1) + 2 (0,1,0) + (1,0,0) = (1,2,1)$. 
The arithmetic frame-work provides automatically an algebraic
structure on $\Delta$-sets with properties we want to have. 

\paragraph{}
Let us look at the case, where $G$ is a connected simplicial complex and $K=\{ \{v_1\},\dots,\{v_m\}\}$ is
a {\bf zero-dimensional skeleton complex}, consisting of finitely many isolated points. Let us start with $m=1$,
where we just remove one point from $G$ and where $U=G \setminus \{ \{v_1\} \}$. The Betti vector of $U$ 
is now the {\bf reduced Betti vector}. The cohomology of $U$ is the {\bf reduced cohomology}. 
It has now concrete representatives as harmonic forms of finite matrices. Even if $G$ and $K$ shuld be huge
and $U$ is small, then the relative cohomology $H(G,K)$ can be computed fast as it does not care about how
$K$ and $G$ look like.  The fusion inequality assures that 
$b(G) = (1,b_1,\dots,b_q) \leq b(K) + b(U) =(1,0,0,\dots,0) + b(U)$. 
Since the fist coordinate does not change and only the second coordinate could change in $U$,
the preservation of Euler characteristic forces $b(U) = (0,b_1,\dots,b_q)$. But when removing
the second vertex, the property $b_0 \geq 0$ forces $b(K) = (2,0,\dots,0)$ and 
$b(U) = (0,b_1-1,b_2,\dots,b_q)$. 

\paragraph{}
In general, $b(K)=(m,0,\dots,0)$  and $b(U)=(0,b_1-m+1,b_2,\dots,b_q)$. We can think of $U$ as $G/K$ in 
which all the finite points are glued together. This is an {\bf orbifold picture} a manifold
in which finitely many points are identified. Take a 2-sphere $G$ with $b(G)=(1,0,1)$ 
for example and take $K=\{ \{n\},\{p\} \}$ where $p,q$ are different points with $b(K)=(2,0,0)$.
Then $b(U)=b(G \setminus K) = (0,1,1)$. The harmonic 1-form lives on the edges away from the 
boundary and winds around the point of compactification,
the harmonic $2$-form is the volume form. The orbifold is a ``doughnut without hole", classically
given as a variety $(x^2+y^2+z^2)^2=x^2+y^2$. We were able to derive the cohomology of this 
variety from knowing the cohomology of the $2$-sphere and the $2$-point set alone and that as a
connected set $b_0=1$. The Euler characteristic $\chi(K)+\chi(U)=\chi(G)$ and that $G \to U$ did not 
affect the volume, implied that $b_1=1$. 

\paragraph{}
Let $K$ is a closed simple one-dimensional path in $G$.
An example is a {\bf knot} $S^1$ embedded in a $3$-sphere $G=S^3$. How large is $b(U)+b(K)-b(G)$? 
In the case of an $S^1$ embedded in $S^3$, this means to investigate $b(U)+(1,1,0,0) - (1,0,0,1) \geq 0$. 
The fusion inequality implies that $b(U) = (0,a-1,a,1)$ with $a \geq 1$.
We can also decompose $G=S^3$ as $U + K$ where $U$ is the open solid 3-torus with $b(U)=(0,0,1,1)$ 
(If the open 3-torus and the 2-torus $(1,2,1,0)$ are fused we get the closed 3-torus $(1,1,0,0)$ 
$K$ is the knot complement. )
Now, for the trivial knot and where $U$ is the open solid torus and $K$ is the closed solid torus
then $G=S^3 = U + K$ and $b(U)+b(K)=(0,0,1,1) + (1,1,0,0) = (1,1,1,1)$ and $b(G)=(1,0,0,1)$. 
For the trivial knot $b(I)=(0,1,1,0)$. What happens for a general knot. What happens if the ambient
space is changed? What happens in the case of links, unions of knots. 

\paragraph{}
From the 5 platonic solids, only the {\bf octahedron} and the
{\bf icosahedron} are 2-dimensional simplicial complexes which are Whitney complexes of graphs. 
The {\bf tetrahedron} as the 2-skeleton complex of the complete graph 
$K_3$ is still a simplicial complex.
The other two solids have no triangles and are as Whitney complexes of graphs just 
one-dimensional simplicial complexes. We can see them however as $\Delta$-sets or
as {\bf CW complexes}. Lets look at a polygon. The {\bf square} can 
be represented as $G=\{\{1\},\{2\},\{3\},\{4\},\{1,2\},\{2,3\},\{3,4\},\{1,4\},\{1,2,3,4\}\}$ but 
now, $\{1,2,3,4\}$ is a 2-dimensional cell, not a 3-dimensional one. When considering $\{1,2,3,4\}$
to be a 3-simplex, the Betti vector of the $\Delta$ set would be $(1,1,0,1)$. With the adapted dimension
function, the Betti vector becomes the expected $(1,0,0)$ as $G$ is topologically a 2-ball.
In order to deal with a general $G$ and a non-orthodox dimension function, we need to replace the
now assumed ${\rm dim}(x) = |x|-1$.

\paragraph{}
A remark to the de Rham picture (squares) relating to simplicial picture (triangles). Still look at
the polygon $G$ which has an addition 2 dimensional element $\{1,2,3,4\}$. 
If we write a vector field on the 1-dimensional part edges as $(P,Q)$, where $P$ is the function on
the horizontal edges ${1,2},{3,4}$ and $Q$ the function on the vertical edges ${2,3},{1,4}$. 
The exterior derivative of the face maps are just $Q_x-P_y$, known in multi-variable calculus 
courses. This example is related to the 
De Rham theorem which relates the cohomology based on rectangular regions with the cohomology based on 
simplices. The relation can be understood using homotopy as the 2-cell $U(\{4\}) =$
$\{ \{4\}, \{1,4\},\{3,4\},\{1,2,3,4\} \}$ is null-homotop leaving $G \setminus U(\{4\}) =$
$\{\{1\},\{2\},\{3\},\{1,2\},\{2,3\}\}$ which can be seen as $K_3 \setminus U(\{2,3\})$, 
with null-homotop $U(\{2,3\})=\{ \{2,3\},\{1,2,3\} \}$. So, the original square is homotop to 
the complete triangle. The cohomology does not change under homotopies. Open sets help here
to understand the chain homotopy which is needed in the de-Rham theorem. 

\section{More remarks}

\paragraph{}
While this is a {\bf Mayer-Vietoris theme} it is close to classical frame works we would like to point out what
is different. Looking at open sets allows to give explicit representation of cohomology classes in the form 
of a basis of kernels of finite matrices. Not only for the closed sets, where Hodge theory is equivalent but
also for open sets $U = G \setminus K$, where we can interpret $U$ as the set $G/K$ with the point representing 
$K$ removed or as relative cohomology $H(G,K)$. We can see an open set as a particular case of a $\Delta$ set
a larger category than simplicial complexes.

\paragraph{}
The theme fits into a {\bf finitist setting}. We have assumed forms to be real values but the field could be
replaced.  Instead of the field $k=\mathbb{R}$ we could take $k=\mathbb{Q}$ and for a given finite $G$
restrict to finite subsets of $\mathbb{Q}$. Row reduction gives rational and so finite harmonic representatives
of the cohomology classes. All matrices could also be considered over other fields like {\bf finite fields} $K$.
The Betti vectors $b(K,k)$ and $b(U,k)$ still can be defined (but have an other meaning of course). 
We did not investigate the relation between $b(K,k),b(U,k)$ and $b(G,k)$ for finite fields,
if the Betti vectors were defined in the same spectral theoretical way. 

\paragraph{}
Cohomology of open sets is a computationally fast path to 
{\bf relative cohomology} that satisfies the {\bf Eilenberg-Steenrod axioms}.
These axioms state to have have compatibility with arithmetic $b(0)=0,b(1)=1, b(A+B)=b(A+B)$, 
excision $b(G+V,K+V)=b(G,K)$ and compatibility with homotopy.
\footnote{Axiom 4 in \cite{EilenbergSteenrod} needs adaptation} 
The {\bf exactness axiom} appears to be related to the fusion in equality.
The cohomology was {\bf defined} using elementary linear algebra, as kernels of fixed matrices 
and not through equivalence classes of cocycles over coboundaries. By Hodge theory, 
this agrees for closed set with the classical frame work. It also works for open sets. 
In a history section, some references are given to the Hodge approach to cohomology. It started
in the 40ies, was then further used in the second half of the 20th century and reemerged in the 
21th century at various applied situations like \cite{DNA2006}. 

\paragraph{}
Just to illustrate the difference, we should point out that classically, 
the excision property needs to be verified. It is a result called the {\bf excision theorem}. 
In the calculus of the cohomology of open sets, excision given as 
Axiom 6 in \cite{EilenbergSteenrod} does not require proof here. 
Excision is a direct consequence of the set theoretical identity $(G \setminus V) \setminus (U \setminus V) 
= G \setminus U$ for open sets $U,V$ with $V \subset U$ in a topological space $(G, \mathcal{O})$. 

\paragraph{}
The Mayer-Vietoris theorem relates the cohomology of $U$ and $V$ with the cohomology 
of $U \cup V$, if $U,V$ are both closed and $U \cap V$ in overlapping situations. We are already 
not aware of a continuum analog of the example situation, where we split a $q$-sphere $G$ into 
an open $U$ and closed ball $K$. The relative cohomology picture $H(G,K)$ sees this
as the punctured sphere $G/K \setminus \{x\}$, where $x$ is the equivalence class of $K$. 
We see $H(G,K)=H(U)$ the cohomology of $G$ with $K$ removed (we have defined what we mean with 
a cohomology of an open set explicitly) and can define
$H(G,U)$ as $H(K)$ with $U$ removed which is more familiar as $K$ is a simplicial complex.
The interpretation $H(G,U)=H(K)$ is again just an interpretation of the excision property 
$H(G\setminus U,U \setminus U) =H(K,\emptyset) = H(K)$.

\paragraph{}
In Mayer-Vietoris situations, a topological space $G$ is written as a
union of two topological spaces $A,B$ which are therefore closed sets with whose overlap $A \cap B$
is again a closed set.  While $b_U+b_K=b_G$ is not always true, we
have seen here that $H(G)$ as the direct sum $H(U) \oplus H(K)$ with some pair identifications
or ``odd-even" cohomologies. We liked to see the disappearance of cohomology when moving from $(H,U)$ to $G$
as ``collision of particles" as physics often associates harmonic forms with particle-like structures. 
An open $k$-ball in the discrete has the cohomology $(0,0,\dots,1)$, 
the simplest case being a single $k$-simplex $x$ forming an open set $U=\{x\}$. 
Its closure $\overline{U}$ is the complete simplicial complex on $(k+1)$ vertices. 
The Hodge Laplacian of this $U$ is the $1 \times 1$ matrix with $0$ entry,
explaining the Betti vector $b_U=(0,0,\dots,1)$. As $K$ is contractible, its Betti vector is
$(1,0,\dots,0)$. 

\paragraph{}
For every closed $K \subset G$ we have {\bf defined} a {\bf relative cohomology} $H(G,K)=H(U)$ and for 
every open $U \subset G$ we have defined a {\bf relative cohomology} $H(G,U)=H(K)$. 
This is notation, but it agrees often with what one understands as a {\bf generalized homology
theory} as defined by Eilenberg and Steenrod in 1942. 
In particular, by definition, the {\bf excision} property $H(G \setminus V,K \setminus V)=H(G,K)$
holds for open sets $V \subset K$ and $H(G \setminus L,U \setminus L)
=H(G,U)$ holds for closed subsets $L \subset U$, reflecting that the complement of $K \setminus V$ 
in $G \setminus V$ is the same than the complement of $K$ in $G$. 
The axioms of {\bf Eilenberg-Steenrod} \cite{EilenbergSteenrod} were the homotopy, exactness, additivity 
and excision. In general $H(G,\{x\})$ is the {\bf reduced cohomology}.
$H(G, G \setminus \{x\})= H(\{x\})$ is the cohomology of a point. 

\paragraph{}
The {\bf quotient} $G/K$ can be defined as the 1-point compactification 
$\dot{U}$ of the open set $U=G-K$, which is a $\Delta$ set and not a simplicial complex in general. 
We can think about $G/K$ as taking $G$ and collapsing all simplices in $K$ to points. 
If we wanted to write down the $\Delta$-set define the zero-dimensional complex $1=\{p\}$. This
already determines that the cohomology just adds $e_1=(1,0,0,...)$. As for the topology, just add
the union $U \cup \{p\}$ as an additional open set. The space $U/K$ now connected as $\{p\}$
is not an open set and $G$ can not be written as a union of disjoint open sets because one of them
must contain $p$ and so must be the entire $G$. 
We $H(U) = H(G,K)=H(G/K \setminus \{p\})$. That $G/K$ can in general not be matched to a simplicial complex
follows from the {\bf Kruskal-Katona} constraints \cite{Lovasz1993,Knuth2011} 
and the fact that open sets do not have such constraints.

\paragraph{}
We can see however $G/K$ as a {\bf $\Delta$-set} which has a topology and cohomology. 
For example, if $K$ is the $(q-1)$-dimensional {\bf skeleton sub complex} of of $q$-sphere, then $G/K$
corresponds to a collection of open balls (the maximal simplices) and $b(G/K)-e_1=b(U)=(0,0,\dots,0,m)$,
where $m$ is the number of maximal simplices. Topologically, $G/K$ would be considered
a {\bf bouquet of $m$ spheres of dimension $q$} and $G/K-\{x\}$ as a collection of $m$ disjoint open 
$q$-balls. We can identify $G/K$ naturally with the {\bf 1-point compactification} 
of $U$ so that $b(G/K) = b(U) + e_1$. In the case of $b_0$ connected components we have 
$b(G/K) = b(U) + b_0 e_1$. This can be iterated. The sphere $G=\{ \{3\}, \{1,2\}  \}$ for 
example can be seen as a quiver with $1$ vertex and $1$ loop if $\{3\}$ is considered a 
{\bf closed set} so that $G$ is not disconnected. 
\footnote{If $\{3\}$ was considered open, we would deal with an open set with a single
open vertex and an open edge.} We see it as a 1-point compactification
of the simplex $G_0=\{1,2\}$, which is a $1$-ball (= open interval). we see in this example
that we might have to say whether given $\Delta$-set is open or closed. We can think of 
$G=\{ \{3\}, \{1,2\}  \}$ as a disconnected open set with two disjoint sets in its closure 
$K_2+K_1$ or then as a closed, connected topological space with three open sets 
$\{ G,\emptyset, \{1,2\} \}$  and and closed sets $\{ G,\emptyset, \{3\} \}$. The later
interpretation with 3 open ses is a $1$-sphere, a 1-point compactification of a 1-ball. 

\paragraph{}
Lets look at two examples leading to loops or multiple connections. If is the cyclic complex $G=C_n$
and $K$ is the zero dimensional skeleton complex in $G$ consisting of $n$ points, then $G/K$ is a
quiver with one vertex and $n$ loops. The open set $U$ consists of $n$ edges and $b(U)=(0,n)$. 
So, $b(G/K)=1+b(U)=(1,n)$ which is the cohomology of a bouquet of $n$ spheres. As a second example, 
lets look at $C_n$, a circular complex of length $n$. Take for $K$ the path complex $K=P_{n-2}$ of length $n-3$
with $n-2$ vertices and $n-3$ edges. Now $U$ is an open set $U(x)$ of a vertex which has $b(U)=(0,1)$. 
The quotient $G/K$ is the $1$-point compactification of $U$ which is a multi-graph with two vertices and
two edges connecting them. Its cohomology is $(1,1)$ and indeed this should be considered a circle. 
If we would have taken a path complex $P_{n-1}$ away, then $U$ would have one edge only again with $b(U)=(0,1)$
and $b(G/U)=(1,1)$, but now $G/U$ is a vertex with a single loop. This should again be seen as a non-simplicial
complex implementation of a circle. In general, the one point compactification of a $q$-simplex $U$ with 
$b(U)=(0,0, \dots, 0,1)$ is a $q$ sphere with $b(\dot{U})=(1,0, \dots, 0,1)$, the smallest $\Delta$ set 
implementation of a sphere. 

\paragraph{}
Open sets are useful also when looking at complexes allowing a {\bf group action}.
Let $A$ be a cyclic subgroup of the {\bf automorphism group} of $G$, let $K$ be
the {\bf fixed point set} of $A$. The complement $U$ is invariant. Assume that $V=U/A$ is the
{\bf fundamental region}, the smallest open set such that $A V = U$. While $U/A$ is open,
the object $G/A$ is not a simplicial complex and not closed. 
Because $\chi(V) = \chi(U/A) = \chi(U)/|A|$, we have the {\bf Riemann-Hurwitz formula} 
$\chi(G) = |A| \chi(G/A) + (|A|-1) \chi(K)$. 

\paragraph{}
This can be generalized to general groups $A$ where $\chi(G) = |A| \chi(G/A) + R$, where $R=\sum_x (e_x-1)$ 
is a ramification part with $e_x=1+\sum_{a \neq 1, a(x)=x} (-1)^{|x|}$. In the zero-dimensional case, 
this is the {\bf Burnside lemma} $|G| + \sum_{a\in A, a \neq 1} \sum_{x \in G, a(x)=x} 1 = |A| |G/A|$. 
For the general case one can apply the Burnside lemma on each of the $k$-dimensional simplices. 
For a group element $a \in A$, the fixed point set $K=G_a$ (also called 
{\bf stabilizer} in group theory is a {\bf closed set}, while the complement $U$ is open.
The open components are permuted around. Allowing an analysis, where we can decompose $G=K+U$ and apply
the analysis separately for $K$ and $U$ allows us to deal with a geometric situation as if it was
a set theoretic situation. To say it in fancy words: the group action on a geometric space (Riemann-Hurwitz)
like the topos of $\Delta$ sets can be dealt with like in the case of group actions on the 
topos of sets (Burnside)

\paragraph{}
A special case is when the group $A$ acts without fixed points on $G$ like for example if $G=H \times A$ or more generally
on a {\bf fibre bundle} with fibres $A$, then $\chi(G) = \chi(H) |A|$ which is the product 
formula for Euler characteristic assuming the group $A$ has the discrete topology and 
so is considered a $0$-dimensional complex.
For example if $A$ is an involution of a $q$-sphere $G$ with fixed point set $K$, a $(q-1)$-sphere
and where $G/A$ is a $q$-ball, we have we have $\chi(G) = 2 \chi(B) + \chi(K) = 2-\chi(K)$.
We see that the Euler characteristic of q-spheres is $1+(-1)^q$. We also see that the Euler characdteristic
of a projective space is half the Euler characteristic of the covering sphere. 

\paragraph{}
Every {\bf $\Delta$-set} $\{S_i\}_{i=0}^q, d_i: S_{i+1} \to S_i$ satisfying $d_i d_j=d_{j-1} d_i$ for $i<j$
can be modeled by repeating the closed pair construction of simplicial complexes. It is just important to keep
track of the topology.  The reason why a general $\Delta$ set can be seen as a quotient is when looking 
at the $G$ as a disjoint sum $G=\sum_i (S_i \times K_{i+1})$ modulo $K$, where $K$ is the complex
of all pairs $(x,\delta_i y) + (d_i x,y)$ for $x \in S_i$ and $(i-1)$ simplices $y$ with face maps $\delta_i$ and $x \in S_i$ 
in $G$ of a $\Delta$-set. Turning things around, we can, instead of using the data structure of Delta sets
use the language of {\bf open sets} in generalized simplicial complexes which includes Cartesian products of
simplicial complexes.  Computing the cohomology is no not
more complicated than computing the cohomology of a simplicial complex and the theorems we see for simplicial 
complexes have analog statements for open sets. While Gauss-Bonnet for example puts the {\bf curvature } 
on vertices for closed sets, in the open case, the curvature is supported on the locally maximal sets.
The cohomology of a connected $\Delta$-set $X=G/K$ is the cohomology 
of the 1-point compactification of  an open set $U=G \setminus K$,  where $G$ is itself the cohomology of an open set
etc.  The goal is to compute the cohomology of a rather {\bf arbitrary $\Delta$ set}  fast.
This already can pay off for the simplest possible $\Delta$ sets which are not simplicial complexes like the 
Cartesian product of complexes or the $\Delta$-sets associated to {\bf quivers}. 
We hope to explore this further elsewhere.

\paragraph{}
We can now work with {\bf topological pairs} $(G,K)$ in which $G$ itself is a topological pair
and where the starting point of a topological pair $(G,K)$ in which $G,K$ are simplicial complexes. 
There is a map $(G,K) \to G/K=\dot{U}$ on topological spaces, where $U=G\setminus K$ and $\dot{U}$
is the 1-point compactification. When taking different such maps, we can realize equivalence classes of 
simplicial complexes. A {\bf $\Delta$ map} is then be simply {\bf continuous map} of the 
topological space representing it. One can see $S_i$ as the set of open sets in $U$ of dimension $i$ and $d_i U$ as
the set of open sets of dimension $(i-1)$ in the boundary $\Delta U=\overline{U} \setminus U$ of $U$, where $\overline{U}$ is 
the smallest simplicial complex containing $U$. In summary, the 1-point compactification of a closed sub $\Delta$-set
in a $\Delta$ set is the same as the $\Delta$-set $G/K$.

\paragraph{}
As an example, the one-point compactification of a single open set 
$U=\{x\} =\{ \{1,2\} \}=(G,K)$ within the simplicial complex $G=\{ \{1\},\{2\},\{1,2\} \}$
is a $\Delta$-set $\dot{U} = \{ \{0\}, \{1,2\} \}$ where the only addition open set added is $\{0,1,2\}$.
It is now closed because $\{0,1,2\}$ is declared to be closed too. Note that if we would declare $U$ itself
closed, then it would as a $\delta$ set just be a 1-point set and so contractible. 
\footnote{Compactifying an open set by just declaring it to be closed is the {\bf silly compactification} 
as a smallest open set $U(x)$ containing a simplex of maximal dimension $x$ in the boundary would also be
closed so that the compactification would be disconnected.} 
Now, $\dot{U}$ is topologically a circle and 
{\bf not a simplicial complex}. As a graph, it is a {\bf quiver} $(V,E)$, where $V$ is a single point
and $E$ is a {\bf single loop} on the point. 
Sometimes, we get a simplicial complex: for the one-point compactification of 
$U=\{ \{1,2\},\{2\},\{2,3\},\{3\},\{3,4\},\{4\}, \}$ we get a 
topological circle and a simplicial complex, namely $G=\{ \{1\}, \{1,2\},\{2\},\{2,3\},\{3\},\{3,1\} \}$. 

\paragraph{}
In graph theory, where one knows the {\bf edge collapse} in a graph. 
Given a graph $(V,E)$ we can look at the Whitney simplicial complex $G$, the set of vertex sets of complete 
subgraphs of the graph. Given an edge $e=(a,b)$ in the graph, it defines the {\bf closed set} 
$K = \{ \{a\},\{b\},\{a,b\} \}$. Now, $G/K$ is most of the time again a simplicial complex and even 
the Whitney complex of a new graph in which the edge $e$ has collapsed by identifying the vertices $a,b$.
Already for the cyclic graph $C_4$, doing an edge collapse produces $C_3$, which is no more the simplicial 
complex of a graph because the Whitney complex of $K_3$ is the complete complex with $3$ elements. 
We can proceed however with $C_3$ as the 1-dimensional skeleton complex of $K_3$. Now, if we do an other
edge collapse, we lose the property of having a simplicial complex. We have the 1-point compactification 
of the open set $U=\{ \{2\},\{1,2\},\{2,3\} \}$ which can be seen as a {\bf multi graph} with two edges connecting
two vertices. Indeed, $b(U)=(0,1)$ and $b(G/U)=(1,1)$. We can look at the edge collapse as a $\Delta$-set 
representing a circle. 

\paragraph{}
One can define a {\bf spectral partial order} on arbitrary symmetric matrices $A,B$ by asking 
$\lambda_k(A) \geq \lambda_k(B)$ if the sequences $\lambda_1(A) \leq \cdots \leq \lambda_m(A)$ 
and $\lambda_1(B) \leq \cdots \leq \lambda_n(B)$ are asked to be ordered ascending and if the sequences are 
{\bf left padded} meaning that 0 entries are inserted on the left of the shorter sequence until the sequences
have the same length.  The spectral partial order can apply for any finite symmetric matrices. If $A=[a]$ is a $1 \times 1$ matrix 
for example, then $A \leq B$ is equivalent that the maximal eigenvalue (the spectral radius) is $\geq a$. 
If $A,B$ have the same size and $A$ is larger or equal than $B$ in the {\bf Loewner order} 
which means that $B-A$ positive definite. 
then $A \leq B$ but the reverse is in general not true.
In our case, we have $L_U \leq L_G$ and $L_U \leq L_G$ in the spectral order, 
but the matrices $L_K,L_G$ are {\bf not} Loewner ordered in general. 

\paragraph{}
We see that the difference $b(I) = b(K)+b(U)-b(G)$ is always the sum of
adjacent pairs $(0,\dots,0,1,1,0,\dots,0)$. At the moment, this is still unproven. 
One way to argue why this is true, start with $U$ is empty and successively add new simplices.
There is a dichotomy when adding a $k$-simplex $x$:  
(I) $b(K) \to b(K) + e_k$ or (II) $b(K) \to b(K) - e_{k-1}$. 
Dual is adding a $k$ simplex $x$ to $U$
$b(U) \to b(U) + e_k$ or $b(U) \to b(U) - e_{k+1}$. 
This means that if we remove a $k$-simplex from $U$ we have 
(I) $b(U) \to b(U) - e_k$ or (II) $b(U) \to b(U) + e_{k+1}$.
To prove the dichotomy, look at the Dirac matrix $D$ of $K$ which has the same kernel than $L$.
Removing $x$ means deleting the column x and
row x from $D$. This affects only the $d_k^*$ block (it can reduce the kernel by 1 or not)
and the $d_{k-1}$ block (it can add 1 to the kernel by 1 or not). Not both
are possible. The same argument works for the Dirac operator of $U$, where
we can either remove a column of the $d_k$ block, possibly removing a kernel dimension
and remove a row of the $d_{k+1}$ block possible removing a kernel there.
In order to verify the still unsettled conjecture that $b_I$ is a sum of 
adjacent blocks $e_j+e_{j+1}$, we would only need to show that the combination 
changing $b_{k-1}(K)$ and $b_{k+1}(U)$ can not occur at the same time. 

\paragraph{}
Given a simplicial complex $G$ and an open set $U$ and its complement $K$,
we can see this as a map from $G \to \{0,1\}$. Not all functions can be
realized as such as $U$ needs to be open and $K$ is closed.  We can however
take a random locally maximal simplex $x$ in $K$ and move it to $U$, where it 
will be locally minimal. Or we can take a locally minimal simplex $x$ in $U$
and move it to $K$, where it will be locally maximal.
This allows to define a Markov stochastic dynamics.  We used this for example 
to investigate for which $K$ the {\bf interface cohomology} $b(I)=b(K)+b(U)-b(G)$ 
has maximal norm. 

\paragraph{}
The inequality tells that if we split a system $G$ into 
two subsystems $K,U$, then can only enlarge the dimension of the harmonic forms.
Fusing two such systems together in general decreases
the total dimension of harmonic k-forms. 
We have mentioned this result already in \cite{HodgeInequality}. 
We noticed early January 2023 that in the manifold case, that we often have
equality when doing the connected sum construction. 

\paragraph{}
Having seen the example, where $K,U$ are nice parts of a manifold glued along
a sphere boundary leading to inequality, one can ask what happens if $K$ 
is a sub-manifold of $G$. When do we have equality? 
We do not have equality in general then: 
take a circular sub-complex $K$ of the octahedron graph $G$. Then $b(U)=(0,0,2)$ and 
$b(K)=(1,1,0)$ and $b(G)= (1,0,1)$. By the way, what happens if
$K$ is a {\bf classical knot}, that is if $K$ is a circular subgraph (simple closed curve)
in a $3$-sphere $G$. As $\chi(K)=\chi(G)=0$, also $\chi(U)=0$. 
We have $b(G)=(1,0,0,1)$ and $b(K)=(1,1,0,0)$. Because of the inequality we have
$b(U)=(0,0,1,1)+(0,a,a,0)$ with $b \geq 0$. 

\section{Some historical context}

\paragraph{}
For a modern account in algebraic topology, see \cite{MunkresAlgebraicTopology}.
The axiomatic approach to homology started with Eilenberg and Steenrod \cite{EilenbergSteenrod} 
in 1945 with work starting in 1942. Around the same time, in 1945, Eilenberg and MacLane initiated 
the language of category theory \cite{EilenbergMacLane1945} allowing to see the concept more abstractly
like building the {\bf category of chain complexes}. 
\cite{EilenbergMaclane} look at {\bf star finite complexes} 
which are abstract simplicial complexes in which very star $U(x)$ intersects 
only with finitely many other stars $U(y)$. Some finiteness condition needs to be present
as a collection of countable set of isolated points, a Cantor set, or a
a Hawaiian earring are not finite. 

\paragraph{}
For {\bf finite objects} like {\bf finite abstract simplicial complexes}, 
or generalizations like {\bf CW complexes} or {\bf $\Delta$ sets} like for example
$\Delta$ sets coming from multi-graphs, it suffices to
combine all exterior derivatives to one nilpotent
$n \times n$ matrix $d$ and define the cohomology groups as the {\bf null spaces} of the 
blocks of $(d+d^*)^2$ and the Betti numbers as their dimensions. This keeps all objects and 
computations finite. In the case of $\Delta$ sets, giving the set of sets $G$, the 
matrix $D$ and the integer-valued dimension function ${\rm dim}$ suffices. Especially for
CW complexes which are inductively defined by attaching $k$-balls to already present 
$(k-1)$-spheres, both the matrix $D$ and the dimension function can not be derived directly
from the set of sets $G$ in general. 

\paragraph{}
The category of {\bf simplicial complexes} today is seen as a {\bf quasi topos}. Its completion
is a topos. Simplicial sets and more generally $\Delta$-sets which is a topos too. Both are presheaves
over a simplex category. To every $k$-dimensional simplex $\Delta$ is attached a set $S_k$. Every 
$\Delta$ set is now a contravariant functor on the strict simplex category with monotone maps as 
morphisms. The contravariant functor property leads to the property $d_i d_j = d_{j-1} d_i, i<j$ for
the face maps. We prefer to see things even more general by looking a set $G$ of non-empty sets for which 
the exterior derivative $d$ satisfies $d^2=0$ and where one has a {\bf dimension functional} relating
$|x|$ with the ${\rm dim}(x)$. In the case when ${\rm dim}(x)=|x|-1$ we deal with {\bf simplicial complexes}
or {\bf open sets} or $1$-point compactifications $\overline{U}$ of open sets. In the case of a 
product of simplicial complexes, ${\rm dim}(x)=|x|-2$. In the simplest case, where $G=H=1=\{1\}$ 
we have $G \times H = \{ \{1_G,1_H\} \}$. 
\footnote{Since $\Delta$-sets are a contravariant functor of a sub-category of the simplex category, they are
more general than simplicial sets. One can get a $\Delta$-set from a simplicial set by applying a forgetful functor
to the simplicial set, ignoring the boundary map conditions. $\Delta$ sets are not only 
{\bf more general} they also have {\bf less complexity}.} 

\paragraph{}
For sub-complexes $K$ of a complex, the Hodge definition of cohomology is equivalent to
the classical situation. The importance of Harmonic functions in the context of cohomology was
stressed in \cite{Hodge1933} but might have started already in 1931. 
On open sets $U$ in $G$ however and $\Delta$ set s more generally, it 
opens more possibilites. An open $d$-dimensional ball $U$ for example has the Betti vector 
$b(U) = (b_0,\dots ,b_d) = (0, \dots, 1)$. It can be seen as a punctured $d$-sphere $G$ but 
as $G \setminus \{x\}$ and not $G \setminus U(x)$ which is a closed $d$-ball. 
The set $U$ is the vessel for a volume form but no other harmonic
form. It can be seen dual to a closed ball harboring only a $0$-form and nothing else. 
An interesting question is whether {\bf Poincar\'e-Duality} can be understood in such a way
that an orientable manifold $G$ can be split $G=U \cup K$ with $b(U)=(b_0,\dots,b_q)$ and
$b(K)=(b_q,\dots,b_0)$. In many examples, we have seen that, we can {\bf split the manifold in half}
so that $b(G)$ is the sum of a $b(K)$ and a mirror part $b(U)$. 

\paragraph{}
According to \cite{Dieudonne1989, Pont1974} it took 30 years to build a rigorous theory that is
applicable to general manifolds and that embodying all the ideas initiated by Poincar\'e and Betti. 
Still following \cite{Dieudonne1989}, it was Herman Weyl in 
1923 who pursued a purely algebraic homology theory. This is nowadays more entrenched in
{\bf algebraic combinatorics} rather than {\bf algebraic topology}, subjects which have
moved apart with respect to abstraction levels. Mathematicians like Hassler Whitney 
still saw the subjects together. Whitney also saw graph theory close to topology. Later in the
20th century graph theory has separated and become primarily a theory of one-dimensional 
simplicial complexes or, or when embedded into surfaces as part of
{\bf topological graph theory}. For a systematic treatment of combinatorial topology, see
\cite{Kozlov}, where however abstract simplicial complexes contain the empty set. 

\paragraph{}
Dehn and Heegard introduced abstract simplicial complexes in 
1907 \cite{DehnHeegaard}. 
Eckmann \cite{Eckmann1944} had noted already that Hodge works in a finite linear algebra setting.  
\footnote{We learned elementary fact from a talk by Beno Eckmann and classical Hodge theory from \cite{Cycon}.} 
Eckmann in 1944 credits his PhD advisor Hopf for the inspiration to look at combinatorial versions only. 
Indeed, both Hopf or Alexandrov used finite combinatorial frame works and 
Alexandrov took finite topological spaces serious \cite{Alexandroff1937},
not bothered by the non-Hausdorff property of these finite topological spaces. He would still use
geometric realizations in \cite{alexandroff} but the concept of purely combinatorial notions called
{\bf unrestricted skeleton complexes} (like page 121 in Volume 1 of \cite{alexandroff}.
Historically interesting is that Felix Hausdorff, one of the pioneers of {\bf set theoretical topology}
against the advise of his friend Pavel Alexandroff, decided to cover topology
in his set theory book from the point of view of metric spaces, which then of course
are Hausdorff \cite{HausdorffIII}.


\paragraph{}
The Dirac operator appeared first in 1928 \cite{Dirac1928}.
The square root $D=d+d^*$ of the Hodge operator $L=d d^*+d^* d$ is definitely the simplest square
root of a Laplacian and does not require Clifford algebras, $D$ is also called 
{\bf abstract Hodge Dirac operator} \cite{LeopardiStern}. 
We like it because it is so accessible and is pure linear algebra \cite{KnillILAS}. 
We have made some historical remarks also in \cite{KnillTopology2023,AmazingWorld}.

\paragraph{}
The Kirchhoff Laplacian $L_0$ appeared first in Kichhoff's work in 1847 \cite{Kirc}.
Eckmann's work \cite{Eckmann1944} and the relation of discrete and continuous difference forms \cite{Dodziuk} in 1976,
The method to compute Betti numbers using the kernels of $L_i$ appears in \cite{Friedman1998}. 
The spectra of Laplacians $L$ were investigated in \cite{DanijelaJost,DanijelaJost2}
and estimates given using the Courant-Hilbert theorem, already used in \cite{Dodziuk} who 
also looks already at the zeta function $\zeta(s) = \sum_{\lambda \neq 0} \lambda^{-s}$. 

\section*{Pictures}

\paragraph{}
The topic can be looked at in various ways. \\
 
{\bf 1)} {\bf Cohomology of open sets}. 
The starting point was to extend cohomology from simplicial complexes $K \subset G$ to open sets $U \subset G$ 
in the finite topology defined by $G$. In this finite non-Hausdorff set-up, we can not look at $U$ as a subset of 
its {\bf completion}, as we do not have a metric. In the example of simplicial complexes, 
we can look at $U$ as part of its {\bf closure} $G=\overline{U}$ which is the smallest simplicial 
complex containing $U$.  If $K$ is the complement of $U$ in $G$, then 
we can see $U$ as the pair $(G,K)=(\overline{U}, \delta U)$. For $U=\{ \{x\}, \{1,2\} \}$ for example,
we see $U = G \setminus \{ \{2\} \}$ as the pair $(G,\{2\})$. 
The set $U=\{ \{ x \},\{1,2,3 \} \}$ is neither open nor closed in its closure. 
When talking about an open set, one could understand it as an open set in its closure. As we have
defined the cohomology, if $U \subset \overline{U} \subset G$ and $\overline{U}$ 
then the cohomology of $U$ does not need to see the embedding in $G$. We can not see the cohomology of $U$
by looking at the cohomology of its closure $\overline{U}$. 
However, we can see $U$ is part of the quotient $G/K=\dot{U}$, where $K=G \setminus U$ and have a relation.
This is the $1$-point compactification of $U$ and in general completely unrelated to the closure $\overline{U}$. 
In the case when $U=G \setminus \{ \{ x_0 \} \}$, the {\bf one-point compactification} is $G$. 
For example, $U=\{ \{1,2\} \}$ has the {\bf simplicial complex closure} 
$\overline{U}= \{ \{1\},\{2\},\{1,2\} \}$ but it is also an  
open set the one-point compactification $\dot{U}=G/K$, with $K=\{ \{1\},\{2\} \}$. This means that $U$ 
is an edge in a quiver with one vertex and one loop. As we defined it, the cohomology does not
depend on how we see $U$ compactified as long as we can see $U$ as the complement of a closed subset
$K$ in a simplicial complex $G$.  \\
 
{\bf 2)} {\bf More general objects}. Unlike the cohomology of open sets, the cohomology of closed sets is classical 
because every closed set is in itself a simplicial complex and the cohomology of a subcomplex $K$ of $G$ does 
not relate to how $K$ is embedded in $G$. 
The finite Alexandrov frame-work with a finite topology $\mathcal{O}$ on $G$
allows to extend cohomology to open sets. The complexity of the theory can be best seen from the computer science
point of view: there are non-specialized programming languages in which a basis for all the cohomology groups
can be obtained in {\bf in a handful of lines of code without invoking any libraries}. We only need 
standard procedures like computing the kernel of a finite matrix. The code listed
below not only does give the dimension of the cohomology groups, but provides a concrete basis 
elements of the harmonic forms representing the cohomology groups. We also can look at concrete {\bf harmonic 
forms} for open sets or harmonic forms for $\Delta$ sets. We also have a rich field of experimentation 
for $\Delta$ sets, and especially simplicial sets or  quivers or multi-graphs.  \\

{\bf 3)} {\bf Manifold cutting}. We looked first at the case when $G$ is a manifold and $K,U$ partition
it into an open and closed set. The prototype is to cut a genus $g$ surface in half and keep the
boundary on one side. In that case, if the surface was orientable, we saw additivity of 
cohomology. In the case when $G$ is non-orientable, this was impossible in general.
For example, if $G$ is a projective plane, when cutting
away an open ball produces a volume form on that part which must disappear when fusing it with 
the closed M\"obius strip. We see a fusion of the $2$-form $f$ on $U$ with a $1$ form $g$ on $K$. 
Eilenberg and Steenrod would call $g$ the boundary form of $f$.  This manifold cutting picture
is close to the {\bf Mayer-Vietoris} frame work. In the later, one would take two closed parts
$A,B$ of the manifold and relate the cohomology of $A$ and $B$ and $G$ and $A \cap B$. 
There is no simple formula for the Betti vectors. If $A,B$ are closed balls glued along a $1$-sphere
$A \cap b$ we have $b(A)+b(B)-b(A \cap B) = (1,0,0) + (1,0,0)-(1,1,0) = (1,-1,0)$ and
$b(G) = (1,0,1)$. A more fancy example is the complex $G$ obtained by taking an open $2n$ ball $U=B_{2n}$ 
in which the boundary $K=S^{2n-1}$ is glued by a continuous map $\phi: S^{2n-1} \to S^n$ to a disjoint $K=S^n$. 
Now $b(G) =b(K)+b(U)= (1,0, \dots, 0,1,0, \dots 0) + (0, \dots,0,0,0, \dots 1) = 
(1,0, \dots, 0,1,0, \dots, 0,1)$. The cup product of the middle cohomology class with it self is now a
multiple of the cohomology class of the volume form and this multiple is the {\bf Hopf invariant} of $\phi$. 
The cases $n=1,2,4,8$ are by a theorem of Adams and Atiyah the only cases with Hopf invariant 1. \\

{\bf 4)} {\bf The laboratory picture}. We can see $G$ is a model for a ``large world" in which we can
do physics. A closed set $K$ is a small ``laboratory" in which we do experiments. It is typical in 
laboratory frame works that we {\bf neglect} a lot of other things. This is modeled here with $U$. 
The experiment could for example be a small quantum mechanical system. By itself, this is modeled by a Hilbert
space and a unitary evolution $U(t)=e^{i h t H}$ on this space so that $u(t) = U(t) u(0)$ solves
the Schr\"odinger equation $u'=i h H u$. The spectral properties of $H$ determine
everything. We as an observer are in $G$. When doing a {\bf measurement}, the system $K$
and its complement $U$ are no more independent. This leads to seemingly severe paradoxa like 
``wave collapses" postulated by the Kopenhagen interpretation.
We have here a situation where can on a spectral level compare the separated system $(K,U)$ with
the joined system $G$. The energies of $(K,U)$ are in general smaller $G$, but not always. We believe that
after a suitable Witten deformation (which does not change the $0$-eigenvalues but can change the other 
eigenvalues), that the spectrum of the separated system $L_{K,U}$ is bounded above by the spectrum of 
the joined system $L_G$. We proved smaller than twice the energies of $G$ which is enough to compare 
the dimensions of harmonic states of $G$ with the dimensions of the harmonic states in the 
laboratory-environment system. \\

{\bf 5)} {\bf Harmonic forms as particles.} 
The {\bf Maxwell equations} $dF=0,d^*F=j$ in vacuum, $j=0$ tell that if $G$ is simply connected,
then the electro magnetic field is given by a $1$-form via 
$F=dA$ and $d^* d A =0$.  This can be reinterpreted as $L_1 A =0$ if $A$ is in a {\bf Coulomb
gauge}, meaning $d^* A=0$ (just add a gradient $A+df$ such that $d^*(A+df) = d^* A + d^* df = d^* A + L_0 f=0$).
In other words, electromagnetic waves in vacuum is directly associated to 
{\bf harmonic 1-forms}. Rephrasing this in a particle frame work, we could say that harmonic
$1$-forms represent {\bf photons}. Taking this as a motivating picture, we can also view
{\bf harmonic p-forms} as {\bf particles}. In this picture, we can ``create some particles"
when going from the fused system $G$ into a separated system $K \cup U$. 
If we trap free particles in a closed box $K$, then new harmonic forms emerge. The harmonic forms 
in $G$ still survive in some sense also in the separated system. They can be located either 
in $K$ or in the complement $U$. \\

{\bf 6)} {\bf Projection picture}. If we split a simplicial complex $G$ with $n$ elements
into an open and closed part $G=U \cup K$ we create
two projection matrices $P,Q$. The first map $P$ projects $\mathbb{R}^n$
the space $V$ of harmonic forms on $U \cup K$. 
The second projection $Q$ projects onto the  space $W$ of harmonic forms on $G$.
If $A$ contains the basis of $V$ as columns and $B$ contains the basis of $W$ as columns,
then $P = A (A^* A)^{-1} A^*$ is the projection from $\mathbb{R}^n$ to the kernel of $V$.
 and $Q = B (B^* B)^{-1} B^*$ is the projection from $\mathbb{R}^n$ to the kernel of $W$. 
The concept of two projections is quite common in statistics as the projection onto a smaller dimensional 
space is a {\bf data fitting} process. 
\cite{AndersonJamesTrapp} study the relation between $P-Q$ and $PQ$. The eigenvalues of $P-Q$ 
different from $0$ and $ \pm 1$ determine the eigenvalues of $PQ$ different from $0$ and 
$1$ and that the eigenvalues $1$ of $P-Q$ is equal to $\sum_p b_p(I)$, a number which is
the nullity of the {\bf interface Laplacian} $L(I)=A^T B B^T A$ which is block diagonal too. 
Now $b_p(K)+b_p(U)-b_p(G)= b_p(I)$ is the nullity of $L_p(I)$, the restriction to $p$-forms.
We have therefore a Hodge picture for the interface cohomology in that $b_p(I)$ could be 
computable by looking at kernels of matrices.\\

{\bf 7)}  {\bf Homotopy picture}.
We call an open set $U$ {\bf contractible open set}, if both the closure $\overline{U}$ and the boundary 
$\delta(U) = \overline{U} \cap K$ with $K=G \setminus  U$ are contractible. This means that the
one-poin compactification $\overline{U}$ is a contractible closed set. One can therefore see 
contractible sets identified with punctured contractible closed sets. The empty set $0$ is an 
example of a contractible open set. An other example is $U=K \setminus L$ where both $K,L$ are 
contractible and $L \subset K$. It immediatly follows that $b(U)=0$ if $U$ is contractible. As no 
cohomology can be exchanged with $K$, the complex $G \setminus U$ has the same cohomology than $G$. 
A transformation $G \to G \setminus U$ is a {\bf homotopy reduction}. The reverse is a homotopy extension. 
A small non-empty example is $U = \{ \{1\},\{1,2 \} \}$ which is null-homotop.
The fact that the cohomology does not change if $U$ is homotop to $0$ can be reformulated as 
$H(G,0) = H(G) = H(K) = H(G,U)$. \\

{\bf 8)} {\bf Quotient space}.  If $K$ is a subcomplex of $G$, we can simply {\bf declare} $G/K$ to be the 
1-point compactification of the open set $U=G \setminus K$. Most of the quotients $G/K$ are not
simplicial complex any more. This classically prompted to extend simplicial complexes
to {\bf simplicial sets} or more generally {\bf $\Delta$-sets} = semi-simplicial sets, which are 
simplicial sets after applying a forgetful operator. $\Delta$ sets are still quite intuitive.
Simplicial sets were introduced in 1950 by Eilenberg-Zilber. Since every $\delta$ set naturally can 
be realized as $G/K$, we get, by looking at pairs of simplicial complexes $G,K \subset G$ a structure that is both
intuitive and also allows to compute some cohomologies fast because $H(U)=H(G \setminus K)=H(G,K)$. 
From a computer science point of view, it is more difficult to work with {\bf equivalence classes} of objects.
Direct implementations with concrete data structures is easier. Instead of working with {\bf cohomology classes}  
we concrete elements of the kernel of concrete matrices and this also applies to quotients. \\

{\bf 9)} {\bf Relative cohomology}.  We have {\bf defined} $H(G,K)$ as $H(U)$. Now, also relative cohomology can 
be dealt with without changing any code.  We have a cohomology that satisfies the {\bf Eilenberg-Steenrod }
axioms. These axioms essentially tell that $b(0)=0,b(1)=1, b(A+B)=b(A)+b(B)$, that $b(K+U)=b(G)$ if $b(U)=0$.
We could also include the inequality $b(G) \leq b(K)+b(U)$. We will also have K\"unneth telling that we have
compatibility with Cartesian products in the form $b(G \times H) = b(G) * b(H)$, where $*$ is the convolution of sequences. 
We can use the set-up to compute cohomology faster. Chose an open ball $U$ for example which could be a single $q$-simplex.
Then $G=\overline{U} = U \cup \{x_0\}$ is a $q$-sphere and $b(G) = b(U)+b(1)=(1,0, \dots, 0,1)$. This 
computation of the cohomology of a sphere is probably the fastest, as we only need to compute the kernel 
of $1 \times 1$ matrix to get the cohomology of $U$ as all other matrices are $0 \times 0$ matrices. 
Mathematically, we represent the $q$ sphere as $G=\{ \{0\}, \{1, \dots, q+1\} \}$ which can be seen as a very 
small $\Delta$-set but which  is not a simplicial complex. The topology $\mathcal{O}$ of $G$ 
only has the open sets $\mathcal{O} = (\{\},\{\{1, \dots, q+1\}\},G)$ and closed sets $(\{\},\{\{0\} \},G)$,
definitely a {\bf locale} a pointless topology. \\

{\bf 10)} Also the {\bf Cartesian product} $G \times H = \{ (x,y), x \in G, y \in H\}$ of complexes is a
a $\Delta$ set. It could be seen as an open set in a simplicial complex. 
But as a $\Delta$-set, we must see 2-point sets sets as the new vertices and have dimension $0$. 
This shifts the cohomology. For example $G=\{\{1\},\{2\},\{1,2\}\}$ and $H=\{\{3\},\{4\},\{3,4\}\}$, then 
$G*H=\{ \{1,3\},\{1,4\}, ....  \}$. If we see this as a $\Delta$-set, where $S_0$ are the sets with 2 elements
we have a Cartesian product for which the K\"unneth formula holds by definition, as we can just take the 
{\bf tensor product} of the harmonic forms. 
The {\bf Barycentric refinement} in the Stanley-Reisner picture is much heavier. It involves much large matrices. 
The lean $\Delta$-set implementation only involves $nm \times nm$ matrices if $G$ has $n$ elements 
and $H$ has $m$ elements. Thus, in order to get the cohomology of the product, take the convolution of 
the two Betti vectors of the factor and shift to the left. 
An additional bonus is that we can get the {\bf harmonic forms} of the cohomologies
of the product as the tensor product of the harmonic forms in the factors. In the context of graphs, we had
taken the Shannon product and in order to get the cohomologies of the product taken the tensor product and then 
applied the divergence $d^*$.

\begin{figure}[!htpb]
\scalebox{0.5}{\includegraphics{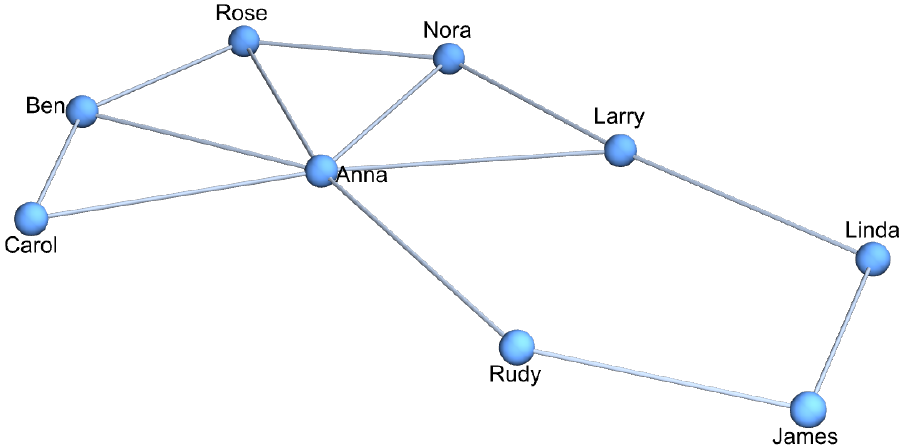}}
\scalebox{0.5}{\includegraphics{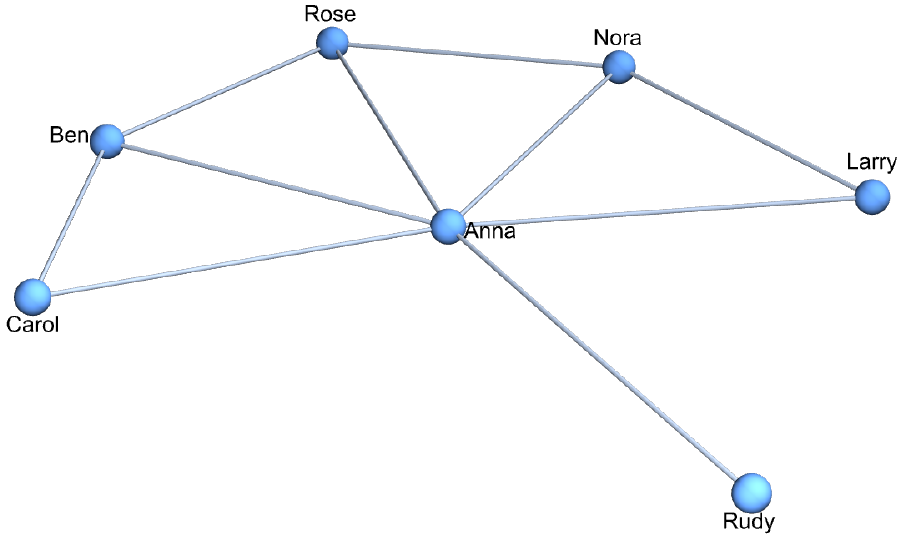}}
\scalebox{0.5}{\includegraphics{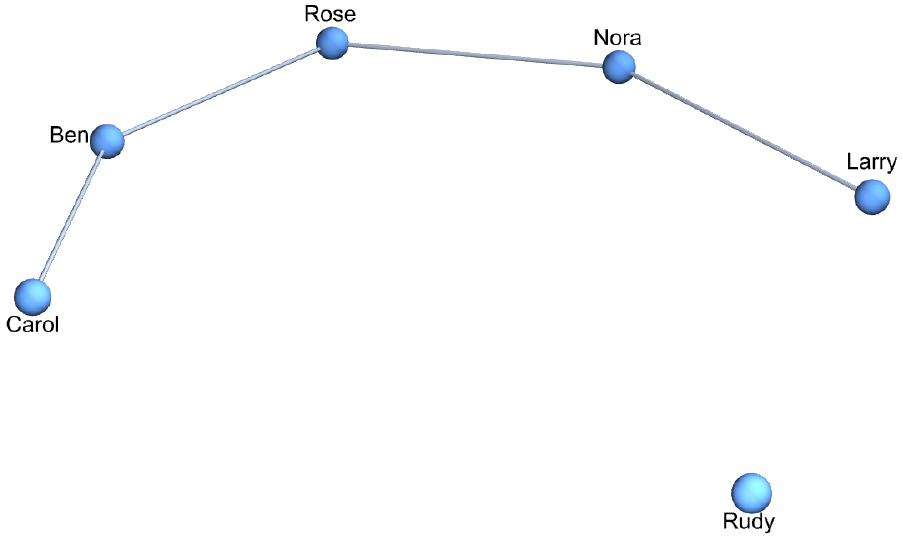}}
\caption{
This graph from the Wolfram example library illustrates a
small network of friends. It defines the simplicial complex $G$ of dimension $1$.
Given the $0$-dimensional point $x=\{ {\rm Anna} \}$ let 
$U(x) = \{ y, x \subset y\}$ its star. The unit ball $B(x)=\overline{U(x)}$ 
is given in the the middle and the unit sphere $S(x) = \delta U$ in the third picture.
All $G,B(x),S(x)$ are closed. The set $U(x)$ has f-vector $(1,6,4)$,
Euler characteristic $1-6+4=-1$ and Betti vector $b(U(x)) = (0,1,0)$.
We have $b(B(x))=(1,0,0)$ and $b(S(x))=(2,0,0)$.
So, $b(U(x)) + b(S(x))=b(B(x))+b(I)$ with $b(I)=(1,1,0)$.
When merging the star $U(x)$ with the unit sphere $S(x)$, a $0$-form
on $S(x)$ and a $1$-form on $U(x)$ have merged.  An other example:
to see that $b(G)=(1,1,0)$, just remove an open set $U=\{ \{ {\rm James,Linda} \} \}$
with $b(U)=(0,1,0)$ to get a contractible $K=G \setminus U$ with $b(K)=(1,0,0)$
so that $b(G)=b(K)+b(U)$. 
}
\end{figure}

\section*{Questions} 

\paragraph{}
A myriad of questions remain non-negative vector $b(I) = b(K)+b(U)-b(G)$. A sample: 
{\bf 1)} Can we characterize the cases where $b_k(I)=0$. What are the properties of the harmonic forms
         which are present in the disjoint union $(K,U)$ and which disappear when fusing the 
         two sets together to $G$?  \\
{\bf 2)} Which $K \subset G$ maximize the functional $\Pi(K) =||b(I)||$? 
         We explored this experimentally using Monte-Carlo simulated annealing (just flip simplices if
         possible from being open or closed when $\Pi$ gets larger).
         It appears as if the maximal $K$ do not look random, if $G$ is a manifold. \\
{\bf 3)} For which $K$ is $b(I)=0$  a minimum $\Pi(K)=0$? This is interesting 
         because having equality allows to compute Betti vectors of fused 
         objects much more easily. \\
{\bf 4)} How does the isospectral deformation of $d_K$ relate to the one of 
         $d_G$ (see \cite{IsospectralDirac2}). Isospectral deformations should be seen 
         as {\bf space symmetries}. They preserve the spectrum of $L$ but $D$ changes
         and so does the Connes distance.  \\
{\bf 5)} A homotopy deformation of $K$ does not change $b(I)$. What happens if $K$ is a submanifold
         of a manifold $G$. Can $b(I)$ produce interesting knot invariants?  \\
{\bf 6)} How do the $b(I)$ fluctuate when applying a Markov process on $K$? 
         The pair $(K,U)$ in $G$ defines a $\{0,1 \}$ valued function on $G$. 
         but we work here with topological constraints in that $K$ needs to be closed at all times. \\
{\bf 7)} Does the cohomology of a submanifold $K$ of a manifold $G$ have any significance in 
         classfying the embeedings? What happens in the co-dimension $2$ case of generalized knots. What
         happens in surgery cases, there the interface identification of $K$ and $U$ is given by 
         a map.  \\
{\bf 8)} Can we cut every orientable manifold $G$ into a closed set $K$ and an open set $U$ such 
         that $b(G)=b(U)+b(K)$ and the reversed sequence of $b(U)$ agrees with $b(K)$? If yes,
         this would provide a different angle to the {\bf Poincar\'e inequality}. \\
{\bf 9)} Is every pair of simplicial complexes $(G,K)$ equivalent to a pair $(G',K')$, where both $G',K'$
         are Whitney complexes of graphs? The answer is yes of course when allowing homotopic examples like
         if $G'$ is the Barycentric refinement of $G$ and $K'$ is the refinement of $K$. 
         For example $C_3$ which is not a Whitney complex is
         equivalent to $(C_4,K_2)$ as $C_4/K_2$ is an edge collapse reducing the cyclic graph with $4$ elements
         to a cyclic graph with $3$ elements without creating a 2-simplex.  \\
{\bf 10)} The fusion inequality can also be looked at for {\bf Wu cohomology} \cite{CohomologyWuCharacteristic}, 
          the cohomology belonging to Wu characteristic $\sum_{x \sim y \in G}  w(x) w(y)$ summing over 
          intersecting pairs of simplices in $G$.  We have observed now that 
          $w(U)+w(K)-w(G)$ can take positive or negative signs. Under which conditions do we
          have an inequality here?

\section*{Appendix: a detailed example}

\paragraph{}
In order to illustrate the set-up and show the difference to traditional frame-works,
let us look at a simple concrete example and write down all the sets and matrices. 
Unlike the usual definition of cohomology as a quotient of ``cocycles" over ``coboundaries",
the Hodge approach is {\bf finite} at every stage and does not require to consider equivalence 
classes in infinite sets. Computationally this is much simpler. All matrices are finite integer matrices. 
Row reduction keeps the matrices rational and an integer basis of the cohomology groups 
can be written down fast and reliably without ever invoking any limit and without ever invoking the 
real numbers. 

\paragraph{}
All associated objects like he set of sets $G$, or the matrices $L(U),L(K),L(G)$ or the 
integers $b(U),b(K),b(G)$ are obtained without invoking the infinity axiom. We do not use geometric 
realizations, but it should be obvious that the Betti vectors of {\bf closed sets}
correspond to the Betti numbers of geometric realizations as they represent 
objects for which Hodge cohomology agrees with the simplicial cohomology. 
Being at all times in a {\bf finitist setting} can be important, not
only in reverse mathematics considerations, or being a shelter in case our axiom systems containing the
infinity axiom system would turn out to be inconsistent but also more elegant from a 
{\bf computer science point of view}.

\paragraph{}
The non-Hausdorff property of the finite topology on a complex already should indicate
that the results for open sets differ from how the cohomology of open sets considered in the 
continuum. Already the simple example coming up indicates how different things are from the continuum
where the cohomology of an open ball is usually also considered to be trivial. In de-Rham set up for
example, every $k$-form $F$ is closed for $k \neq 0$ (for example: every vector field is a
gradient field by Stokes theorem) and a cocycle $0$-form is a constant function.

\paragraph{}
Our example is the discrete 2-sphere 
\begin{tiny}
$$G=\{\{1\},\{2\},\{3\},\{4\},\{1,2\},\{1,3\},\{1,4\},\{2,3\},\{2,4\},
 \{3,4\},\{1,2,3\},\{1,2,4\},\{1,3,4\},\{2,3,4\}\} $$
\end{tiny}
It is a finite set of sets which is closed under the operation of taking finite non-empty subsets
and so a finite abstract simplicial complex. 
By writing down the sets, we already make a choice of order, both
how each elements in a set are ordered and how the sets are ordered. 
It should be obvious that switching the permutation order of the set is realized 
by orthogonal permutation matrices in $\mathbb{R}^n$. Switching an order on 
the simplices will change the exterior derivative matrix $d$ and also be given 
by an orthogonal coordinate change. The basis of the topology has only 14 elements 
$\mathcal{B} = \{ U(x), x \in G \}$ because $G$ has $14$ elements. 
The topology $\mathcal{O}$ itself contains 470 sets. 

\paragraph{}
The closure $B(x) = \overline{U(x)}$ of a smallest open set $U(x)$ is the {\bf unit ball}. Its boundary 
$S(x) = B(x) \setminus U(x)$ is the called the {\bf unit sphere} of $x$. 
Each of these unit spheres is a cyclic complex with 3 or 4 elements and so a 1-sphere, 
a cylic complex. For example
\begin{tiny}
$U(\{1\}) = \{\{1\},\{1,2\},\{1,3\},\{1,4\},\{1,2,3\},\{1,2,4\},\{1,3,4\}\}$, \\
$B(\{1\}) = \{\{1\},\{2\},\{3\},\{4\},\{1,2\},\{1,3\},\{1,4\},\{2,3\},\{2,4\},\{3,4\},\{1,2,3\},\{1,2,4\},\{1,3,4\}\}$ \\
$S(\{1\}) = \{\{2\},\{3\},\{4\},\{2,3\},\{2,4\},\{3,4\}\}$.
\end{tiny}

\paragraph{}
Because $G \setminus U(x) = $ is contractible, $G$ is a 2-sphere. 
For example, $G \setminus U(\{1\}) = \{\{2\},\{3\},\{4\},\{2,3\},\{2,4\},\{3,4\},\{2,3,4\}\}$ is the
Whitney complex of a complete graph $K_3$. A cyclic complex is a $1$-sphere because
every unit sphere there consists of $2$ separated points and so is a $0$-sphere and taking a point away
from a $1$-sphere produces a contractible complex. Let now $K=B(x)$ with $x=\{ 1,2,3 \}$ and $U$ its
complement. Written out, this means
\begin{tiny}
$$ K=\{\{1\},\{2\},\{3\},\{1,2\},\{1,3\},\{2,3\},\{1,2,3\}\}, 
   U=\{\{4\},\{1,4\},\{2,4\},\{3,4\},\{1,2,4\},\{1,3,4\},\{2,3,4\}\} \; . $$
\end{tiny}
The matrices $d$ can be read off as the lower diagonal part of the Dirac matrix $D=d+d^*$. 
For $U$ and $K$ the Dirac matrices are 
\begin{tiny}
$$ D(U) = \left[ \begin{array}{ccccccc}
                   0 & 0 & 0 & -1 & -1 & 0 & 0 \\
                   0 & 0 & 0 & 1 & 0 & -1 & 0 \\
                   0 & 0 & 0 & 0 & 1 & 1 & 0 \\
                   -1 & 1 & 0 & 0 & 0 & 0 & 1 \\
                   -1 & 0 & 1 & 0 & 0 & 0 & -1 \\
                   0 & -1 & 1 & 0 & 0 & 0 & 1 \\
                   0 & 0 & 0 & 1 & -1 & 1 & 0 \\
                  \end{array} \right], 
   D(K) = \left[ \begin{array}{ccccccc}
                   0 & 1 & 1 & 1 & 0 & 0 & 0 \\
                   1 & 0 & 0 & 0 & -1 & -1 & 0 \\
                   1 & 0 & 0 & 0 & 1 & 0 & -1 \\
                   1 & 0 & 0 & 0 & 0 & 1 & 1 \\
                   0 & -1 & 1 & 0 & 0 & 0 & 0 \\
                   0 & -1 & 0 & 1 & 0 & 0 & 0 \\
                   0 & 0 & -1 & 1 & 0 & 0 & 0 \\
                  \end{array} \right] \; . $$
\end{tiny}
The square are the Hodge Laplacians:
\begin{tiny}
$$ 
   L(U) = \left[ \begin{array}{ccccccc}
                   3 & 0 & 0 & 0 & 0 & 0 & 0 \\
                   0 & 3 & 0 & 0 & 0 & 0 & 0 \\
                   0 & 0 & 3 & 0 & 0 & 0 & 0 \\
                   0 & 0 & 0 & 3 & 0 & 0 & 0 \\
                   0 & 0 & 0 & 0 & 2 & 1 & -1 \\
                   0 & 0 & 0 & 0 & 1 & 2 & 1 \\
                   0 & 0 & 0 & 0 & -1 & 1 & 2 \\
                  \end{array} \right]
   L(K) = \left[
                  \begin{array}{ccccccc}
                   2 & -1 & -1 & 0 & 0 & 0 & 0 \\
                   -1 & 2 & -1 & 0 & 0 & 0 & 0 \\
                   -1 & -1 & 2 & 0 & 0 & 0 & 0 \\
                   0 & 0 & 0 & 3 & 0 & 0 & 0 \\
                   0 & 0 & 0 & 0 & 3 & 0 & 0 \\
                   0 & 0 & 0 & 0 & 0 & 3 & 0 \\
                   0 & 0 & 0 & 0 & 0 & 0 & 3 \\
                  \end{array}
                  \right] \; . $$
\end{tiny} 
By looking at the kernels of the blocks (an $1 \times 1$, 
a $3 \times 3$ and a $3 \times 3$ block for $U$ and a
$3 \times 3$, a $3 \times 3$ and a $1 \times 1$ block for $K$),
we see that $b(U)=(0,0,1)$ and $b(K)=(1,0,0)$.  This is an example
where have a {\bf Poincar\'e-duality splitting}.

\paragraph{}
Finally, let us look at the Hodge matrix of the complex $G$ itself:

\begin{tiny}
$$ L(G) = \left[ \begin{array}{cccccccccccccc}
                   3 & -1 & -1 & -1 & 0 & 0 & 0 & 0 & 0 & 0 & 0 & 0 & 0 & 0 \\
                   -1 & 3 & -1 & -1 & 0 & 0 & 0 & 0 & 0 & 0 & 0 & 0 & 0 & 0 \\
                   -1 & -1 & 3 & -1 & 0 & 0 & 0 & 0 & 0 & 0 & 0 & 0 & 0 & 0 \\
                   -1 & -1 & -1 & 3 & 0 & 0 & 0 & 0 & 0 & 0 & 0 & 0 & 0 & 0 \\
                   0 & 0 & 0 & 0 & 4 & 0 & 0 & 0 & 0 & 0 & 0 & 0 & 0 & 0 \\
                   0 & 0 & 0 & 0 & 0 & 4 & 0 & 0 & 0 & 0 & 0 & 0 & 0 & 0 \\
                   0 & 0 & 0 & 0 & 0 & 0 & 4 & 0 & 0 & 0 & 0 & 0 & 0 & 0 \\
                   0 & 0 & 0 & 0 & 0 & 0 & 0 & 4 & 0 & 0 & 0 & 0 & 0 & 0 \\
                   0 & 0 & 0 & 0 & 0 & 0 & 0 & 0 & 4 & 0 & 0 & 0 & 0 & 0 \\
                   0 & 0 & 0 & 0 & 0 & 0 & 0 & 0 & 0 & 4 & 0 & 0 & 0 & 0 \\
                   0 & 0 & 0 & 0 & 0 & 0 & 0 & 0 & 0 & 0 & 3 & 1 & -1 & 1 \\
                   0 & 0 & 0 & 0 & 0 & 0 & 0 & 0 & 0 & 0 & 1 & 3 & 1 & -1 \\
                   0 & 0 & 0 & 0 & 0 & 0 & 0 & 0 & 0 & 0 & -1 & 1 & 3 & 1 \\
                   0 & 0 & 0 & 0 & 0 & 0 & 0 & 0 & 0 & 0 & 1 & -1 & 1 & 3 \\
                  \end{array} \right]  \; .  $$
\end{tiny}
We see three blocks, where the first and last both have a $1$-dimensional 
kernel and where the middle block does have a trivial kernel. The Betti vector is 
$b(G) = (1,0,1)$. In this example, we have $b(K) + b(U) = b(G)$. 

\paragraph{}
This happens in all dimensions: the $q$-sphere is the disjoint union of an 
open ball and a closed ball. The closed ball $K$ has $b(K)=(1,0,0, \dots,0)$, 
the open ball $U$ has $b(U) = (0, \dots , 0,1)$. The sphere $G$ has 
$b(G)=(1,0,\dots, 0,1)$. 

\section*{Appendix: Code}

\paragraph{}
We again include the self-contained Mathematica code. It can be copy-pasted 
from the LaTeX source code of this document on the ArXiv. This allows an interested reader to 
experiment and compute the Betti vectors for an arbitrary pair of complementary $K,U$ in 
a complex $G$. We also implemented the Cartesian product but did not yet shift the Betti
vector. 

\paragraph{}
The code to compute the cohomology of simplicial complex does not have to be changed get the cohomology of a
simplicial set. We just also need to pay attention to the dimension functional. In the
code below, we have already implemented that when computing the Dirac operator, we also include the dimension markers.
So, the data structure of a $\Delta$-set is given by $(G,D,r)$ where $r=(r_0,r_1, \dots, r_q)$ are the places
where the dimensions change. For example, if we look at $U=\{ \{ 1,2 \} \}$, then the Dirac matrix is the
$1 \times 1$ matrix $[0]$ and $r=(0,0,1)$ can be deduced from $U$. 
We have $b(U) = (0,1)$. Run $U=\{\{1,2\}\}$  $Dirac[U]$ and $Betti[U]$ to see this. 
In the case of the {\bf $1$-point compactification} $\dot{U} = \{ \{ 3 \}, \{1,2\} \}$ we have 
$D$ is the $2 \times 2$ matrix $ \left[ \begin{array}{cc} 0 & 0 \\ 0 & 0 \\ \end{array} \right]$ and $r=(0,1,2)$. 
The Betti vector is $b(\dot{U}) = (1,1)$ as it should be as the 1-point compactification of a 
single $k$-simplex is a $k$-sphere. 
As for the {\bf compactification} $G=\overline{U}=\{ \{ 1\},\{2\},\{1,2\} \}$, we have
the Dirac matrix $D=\left[ \begin{array}{ccc} 0 & 0 & -1 \\ 0 & 0 & 1 \\ -1 & 1 & 0 \\ \end{array} \right]$ and $r=(0,2,3)$
(as we have $f_0=2-0$ vertices and $f_1=3-2$ edges) and $b(G) = (1,0)$. You can get this by typing $Dirac[Cl[\{\{1,2\}\}]]$. 

\paragraph{}
Here is how one can compute the cohomology of a higher dimensional tori in an effective way. 
Start with the 1-torus $S=\dot{U}=\{ 3\},\{1,2\} \}$. Now look at $S \times S$ (implemented in our
code as $Shannon[S,S]$) which gives $S \times S = \{\{3,6\},\{1,2,6\},\{3,4,5\},\{1,2,4,5\}\}$ a
set of sets with 4 elements! The Dirac matrix $D$ is the $4 \times 4$ matrix $0$ and $r=(0,0,1,3,4)$ 
needs to be adapted to $(0,1,3,4)$ in order to adjust for the dimension shift as the zero dimensional
points now are represented as sets of cardinality $2$. The Betti vector without modification is $(0,1,2,1)$
but after adjustment becomes $(1,2,1)$. We can see this as dividing by $1$ because $1 \times 1/1$ is $1$
and multiplication by $1$ shifts the dimension function. Lets continue and look at 
\begin{tiny}
$S \times S \times S = \{\{3,6,9\},\{1,2,6,9\},\{3,4,5,9\},\{3,6,7,8\},\{1,2,4,5,9\},\{1,2,6,7,8\},
    \{3,4,5,7,8\},\{1,2,4,5,7,8\}\}$ 
\end{tiny}
a set of sets with $8$ elements representing a 3-dimensional torus. 
The Dirac matrix is a $8 \times 8$ matrix containing all zeros. The unadjusted dimension markers
are $r=(0,0,0,1,4,7,8)$. When adjusted (divide by $1 \times 1$) we get $r=(0,1,4,7,8)$ and Betti
vector $(1,3,3,1)$ as it should be for a $3$-torus. 

\paragraph{}
This is a dramatic computational advantage. What we have done previously using the Stanley-Reisner computation 
is to to produce a simplicial complex (the order complex) which has the effect that $G \times 1$ is the 
{\bf Barycentric refinement} of $G$. If we insist on simplicial complexes, then the smallest simplicial complex
representing a sphere is $G=\overline{\{\{1,2\},\{1,3\},\{2,3\}\}}$. The Stanley-Reisner picture product
$G \times_1 G$ (the Barycentric refinement of $G \times G$ as a $\Delta$ set) has now 216 elements and
the complex $G \times_1 G \times_1 G$ has 33696 elements. Computing the Betti vector requires to look at a 
$33696 \times 33696$ matrix.  An other example is to model the 4-manifold $G=S^2 \times S^2$ as $H \times H$
with $H=\{1,\{2,3,4\} \}$ which is the $1$-point compactification of $U=\{ \{1,2,3\} \}$. 
Now, $b(G) = (1,0,2,0,1)$ reflects the fact that we have besides the constant $0$
form and the volume form also two $2$-forms located on the two factors. This of course is just K\"unneth. In 
the present frame work K\"unneth is almost trivial as we just take the tensor product $f \star g(x,y) = f(x) g(y)$
of two harmonic forms $f$ on $A$ and $g$ on $G$ to get the harmonic form $f \star g$ on $A \times B$. Already 
Whitney had been battling with the fact that the product of a $k$ simplex and $l$ simples is a $k+l+1$ simplex
rather than a $k+l$ simplex. The concept of $\Delta$ sets (or allowing to shift the dimension function makes
this problem go away. We had previously dealt with this by applying the divergence to the naive tensor product
of a $p$ form and $q$ form in order to get a $p+q$-form. 

\vfill 
\pagebreak

\begin{tiny}
\lstset{language=Mathematica} \lstset{frameround=fttt}
\begin{lstlisting}[frame=single]
F[G_]:=Module[{l=Map[Length,G]},If[G=={},{},
 Table[Sum[If[l[[j]]==k,1,0],{j,Length[l]}],{k,Max[l]}]]]; s[x_]:=Signature[x];L=Length;
s[x_,y_]:=If[SubsetQ[x,y]&&(L[x]==L[y]+1),s[Prepend[y,Complement[x,y][[1]]]]*s[x],0];
Dirac[G_]:=Module[{f=F[G],b,d,n=Length[G]},b=Prepend[Table[Sum[f[[l]],{l,k}],{k,Length[f]}],0];
 d=Table[s[G[[i]],G[[j]]],{i,n},{j,n}]; {d+Transpose[d],b}];
Hodge[G_]:=Module[{Q,b,H},{Q,b}=Dirac[G];H=Q.Q;Table[Table[H[[b[[k]]+i,b[[k]]+j]],
 {i,b[[k+1]]-b[[k]]},{j,b[[k+1]]-b[[k]]}],{k,Length[b]-1}]];       Betti[G_]:=Map[nu,Hodge[G]];
Closure[A_]:=If[A=={},{},Delete[Union[Sort[Flatten[Map[Subsets,A],1]]],1]];
Whitney[s_]:=If[Length[EdgeList[s]]==0,Map[{#}&,VertexList[s]], Map[Sort,Sort[Closure[
  FindClique[s,Infinity,All]]]]];                    nu[A_]:=If[A=={},0,Length[NullSpace[A]]];
Shannon[A_,B_]:=Module[{q=Max[Flatten[A]],Q,G={}}, Q=Table[B[[k]]+q,{k,Length[B]}];
  Do[G=Append[G,Sort[Union[A[[a]],Q[[b]]]]],{a,Length[A]},{b,Length[Q]}]; Sort[G]];
OpenStar[G_,x_]:=Module[{U={}},Do[If[SubsetQ[G[[k]],x],U=Append[U,G[[k]]]],{k,Length[G]}];U];
Basis[G_]:=Table[OpenStar[G,G[[k]]],{k,Length[G]}]; Stars=Basis;
RandomOpenSet[G_,k_]:=Module[{A=RandomChoice[Basis[G],k],U={}},Do[U=Union[U,A[[j]]],{j,k}];U];
Betti[G_,U_,K_]:={"b_G=",Betti[G],"b_U=",Betti[U],"b_K=",Betti[K]};
G=Whitney[RandomGraph[{20,50}]]; U=RandomOpenSet[G,10];K=Complement[G,U]; Print[Betti[G,U,K]]; 
S={{1},{2,3}}; G=Shannon[Shannon[S,S],Shannon[S,S]]; Print["Betti(T^4) unshifted",Betti[G] ]
S2={{1},{2,3,4}}; G=Shannon[S2,S2]; Print["Betti S^2 x S^2 (not yet shifted)", Betti[G]] 
\end{lstlisting}
\end{tiny}

\section*{Appendix: Spectral partial order}

\paragraph{}
Define a {\bf partial order on finite sequences} as $x \leq y$ if when ordered and left padded,
one has $x_k \leq y_k$. For example $(1,2,4) \leq (0,0,2,3,5)$ or $(-1,2,3,5) \leq (4,6)$. 
Two sequence which are considered {\bf the isomorphic} if when padded left they are the same. 
Therefore,  $x \leq y$ and $y \leq x$ implies $x=y$. We obviously have $x \leq x$. Finally, there is transitivity
if $x \leq y$ and $y \leq z$, then $x \leq z$. We indeed have a partial order. 

\paragraph{}
Having a partial order on finite sequences, we 
also get a partial order on equivalence classes of symmetric matrices. Given
two matrices $A,B$, which do not necessarily have to be the same size.
We say $A \leq B$ if the eigenvalue sets satisfy $\sigma(A) \leq \sigma(B)$ in 
the partial order of finite sequences. 
Also here, if {\bf similar padded matrices} are identified, then we have a partial 
order on the space of equivalence classes of symmetric matrices. For example, if $A,B$ 
have the same sign and satisfy $A \leq_L B$ in the Loewner order, then $A \leq B$. The 
reverse is not true. The spectral partial order on symmetric matrices
does not require the matrices to have the same size. Let us rephrase what we have
proven for the Laplacians $L_K,L_G$ of two finite abstract simplicial complexes $K,G$
\cite{Eigenvaluebounds}:

\begin{coro} If $K \subset G$ then $L_K \leq L_G$. \end{coro}

\paragraph{}
For two finite sequences $x,y$, define
$x+y$ as the sum of the left padded versions. For example 
$(1,2,3) + (3,7,8,9)= (0,1,2,3) + (3,7,8,9)=(3,8,10,12)$. 
Define $x \oplus y$ as the ordered sequence obtained from the union $x \cup y$. 
For example, $(1,2,3) \oplus (3,7,8,9) = (1,2,3,3,7,8,9)$. 
On the set of all square matrices, define $A+B$ as the sum of left upper padded matrices
and $A \oplus B$ as the direct sum, a matrix which has $A,B$ as blocks. 
The spectra satisfy $\sigma(A+B) \leq \sigma(A) + \sigma(B)$ and
$\sigma(A \oplus B) = \sigma(A) \oplus \sigma(B)$. 

\begin{propo}
Here are some obvious properties for finite ordered sequences $x,y,u,v$ \\
a) If $x \leq y$ and $u \leq v$, then $x+u \leq y+v$.  \\
b) If $x \leq y$ and $u \leq v$, then $x\oplus u \leq y \oplus v$.  \\
c) $x + y \leq 2( x \oplus y)$. \\ 
And the analog for symmetric matrices $A,B,C,D$  \\
d) If $A \leq C$ and $B \leq D$, then $A+B \leq C+D$. \\
e) If $A \leq C$ and $B \leq D$, then $A \oplus B \leq C \oplus D$. \\
f) $A + B \leq 2(A \oplus  B)$.  \\
\end{propo}

\bibliographystyle{plain}

\end{document}